\documentclass[12pt]{amsart}

\copyrightinfo{2000}{American Mathematical Society}

\newcommand{\cA}{{\mathcal A}}

\newcommand{\cD}{{\mathcal D}}
\newcommand{\cE}{{\mathcal E}}
\newcommand{\cL}{{\mathcal L}}

\newcommand{\bN}{{\mathbb N}}
\newcommand{\bZ}{{\mathbb Z}}
\newcommand{\bQ}{{\mathbb Q}}
\newcommand{\bR}{{\mathbb R}}
\newcommand{\bC}{{\mathbb C}}

\numberwithin{equation}{section}

\newtheorem{Theorem}{Theorem}[section]
\newtheorem{Lemma}{Lemma}[section]
\newtheorem{Corollary}{Corollary}[section]

\newtheorem{Remark}{Remark}[section]
\newtheorem{Example}{Example}[section]
\newtheorem{Proposition}{Proposition}[section]

\author{A.~Yu.~Khrennikov}
\address{International Center for Mathematical Modelling in Physics
and Cognitive Sciences MSI, V\"axj\"o University, \ SE-351 95,
V\"axj\"o, \ Sweden. \ Phone: +46\,(470)\,708790 \, Fax: +46\,(470)\,84004}
\email{andrei.khrennikov@msi.vxu.se}

\author{V.~M.~Shelkovich}
\address{Department of Mathematics, St.-Petersburg State Architecture
and Civil Engineering University, \ 2 Krasnoarmeiskaya 4, 190005,
St. Petersburg, \ Russia. \ Phone: +7\,(812)\,2517549 \, Fax: +7\,(812)\,3165872}
\email{shelkv@vs1567.spb.edu}

\title[Non-Haar $p$-adic wavelets and their application]
{Non-Haar $p$-adic wavelets and their application to
pseudo-differential operators and equations}

\thanks{The second author (V.~S.) was supported in part by
DFG Project 436 RUS 113/809, DFG Project 436 RUS 113/951, and
Grant 05-01-04002-NNIOa of Russian Foundation for Basic Research.}



\date{ }

\begin{document}

\begin{abstract}
In this paper a countable family of new compactly supported {\em non-Haar}
$p$-adic wavelet bases in ${\cL}^2(\bQ_p^n)$ is constructed.
We use the wavelet bases in the following applications: in the theory
of $p$-adic pseudo-differential operators and equations.
Namely, we study the connections between wavelet analysis and spectral
analysis of $p$-adic pseudo-differential operators. A criterion for a
multidimensional $p$-adic wavelet to be an eigenfunction for a
pseudo-differential operator is derived. We prove that these
wavelets are eigenfunctions of the fractional operator.
In addition, $p$-adic wavelets are used to construct
solutions of linear and semi-linear pseudo-differential
equations.
Since many $p$-adic models use pseudo-differential operators
(fractional operator), these results can be intensively used in
these models.
\end{abstract}

\maketitle

\section{Introduction}
\label{s1}

According to the well-known Ostrovsky theorem, there are two equal in
rights ``universes'': the real ``universe'' and the $p$-adic one.
The real ``universe'' is based on the field $\bR$ of real numbers,
which is defined as the completion of the field of rational numbers
$\bQ$ with respect to the usual Euclidean distance between rational
numbers. In its turn, the $p$-adic ``universe'' is based on the field
$\bQ_p$ of $p$-adic numbers, which is defined as the completion
of $\bQ$ with respect to the $p$-adic norm $|\cdot|_p$. This norm
is defined as follows: if an arbitrary rational number $x\ne 0$ is
represented as $x=p^{\gamma}\frac{m}{n}$, where $\gamma=\gamma(x)\in \bZ$
and the integers $m$, $n$ are not divisible by $p$, then
\begin{equation}
\label{1}
|x|_p=p^{-\gamma}, \quad x\ne 0, \qquad |0|_p=0.
\end{equation}
The norm $|\cdot|_p$ satisfies the strong triangle inequality
$|x+y|_p\le \max(|x|_p,|y|_p)$ and is {\em non-Archimedean}.

During a few hundred years theoretical physics has been developed
on the basis of real (and later also complex) numbers. However, in the last 20
years the field of $p$-adic numbers ${\bQ}_p$ (as well as its
algebraic extensions) has been intensively used in theoretical
and mathematical physics, stochastics, cognitive sciences and
psychology~\cite{Ar-Dr-V},~\cite{Av-Bik-Koz-O},~\cite{Kh1}--~\cite{Kh4},
~\cite{Kh-Koz3},~\cite{Koch3},~\cite{Koz-Os-Av-1},~\cite{Vl-V-Z}--~\cite{V2}
(see also the references therein). However, since a $p$-adics is a young area
there are many unsufficiently studied problems which have been intensively
studied in the  real ``universe''.  One of them is the $p$-adic wavelet theory.

Nowadays it is difficult to find an engineering area where wavelets
(in the real setting) are not applied. There is a general scheme for the
construction of wavelets in the real setting, which was developed
in the early nineties. This scheme is based on the notion of the
{\em multiresolution analysis} introduced by Y.~Meyer and
S.~Mallat~\cite{Mallat-1},~\cite{Meyer-1}. The $p$-adic wavelet theory
is now in conceptual stage of investigation.
In this theory the situation is as follows.

In 2002, S.~V.~Kozyrev~\cite{Koz0} found a compactly supported
$p$-adic wavelet basis for ${\cL}^2(\bQ_p)$ which is an analog of
the real Haar basis:
\begin{equation}
\label{62.0-1}
\theta_{k;j a}(x)=p^{-j/2}\chi_p\big(p^{-1}k(p^{j}x-a)\big)
\Omega\big(|p^{j}x-a|_p\big), \quad x\in \bQ_p,
\end{equation}
$k=1,2,\dots,p-1$, $j\in \bZ$, $a\in I_p=\bQ_p/\bZ_p$,
where $\Omega(t)$ is the characteristic
function of the segment $[0,1]\subset\bR$,
the function $\chi_p(\xi x)$ is an {\em additive character} of the field $\bQ_p$
for every fixed $\xi \in \bQ_p$ (see Sec.~\ref{s2}).
Kozyrev's wavelet basis (\ref{62.0-1}) is generated by dilatations and
translations of the wavelet functions:
\begin{equation}
\label{62.0-1-0}
\theta_{k}(x)=\chi_p(p^{-1}kx)\Omega(x|_p), \quad x\in \bQ_p, \quad k=1,2,\dots,p-1.
\end{equation}
Multidimensional $p$-adic bases obtained by direct multiplying
out the wavelets (\ref{62.0-1}) were considered in~\cite{Al-Kh-Sh3}.
The Haar wavelet basis (\ref{62.0-1}) was extended to
the ultrametric spaces in~\cite{Kh-Koz1},~\cite{Kh-Koz2},~\cite{Koz2}.

J.~J.~Benedetto and R.~L.~Benedetto~\cite{Ben-Ben},
R.~L.~Benedetto~\cite{Ben1} suggested a method for finding wavelet
bases on the locally compact abelian groups with compact open
subgroups, which includes the $p$-adic setting. They did not develop
the {\em multiresolution analysis {\rm(}MRA{\rm)}}, their method
being based on the {\em theory of wavelet sets}. Their method only
allows the construction of wavelet functions whose
Fourier transforms are the characteristic functions
of some sets (see~\cite[Proposition~5.1.]{Ben-Ben}).
Note that Kozyrev's wavelet basis (\ref{62.0-1}) can be constructed
in the framework of Benedettos' approach~\cite[5.1.]{Ben-Ben}.

The notion of $p$-adic MRA was introduced and a general scheme for
its construction was described in~\cite{S-Sk-1}.
To construct a $p$-adic analog of a classical MRA we need a proper $p$-adic
{\em refinement equation}. In our preprint~\cite{Kh-Sh1}, the following
conjecture was proposed: the equality
\begin{equation}
\label{62.0-3**}
\phi(x)=\sum_{r=0}^{p-1}\phi\Big(\frac{1}{p}x-\frac{r}{p}\Big),
\quad x\in \bQ_p,
\end{equation}
can be considered as a {\em refinement equation}. A solution $\phi$ to this
equation ({\it a refinable function}) is the characteristic function of the unit disc
\begin{equation}
\label{62.0-3**-1}
\phi(x)=\Omega\big(|x|_p\big), \quad x\in \bQ_p.
\end{equation}
The equation (\ref{62.0-3**}) reflects {\it natural} ``self-similarity'' of
the space $\bQ_p$: the unit disc $B_{0}(0)=\{x: |x|_p \le 1\}$ is represented
by a sum of $p$ mutually {\it disjoint} discs
$$
B_{0}(0)=B_{-1}(0)\cup\Big(\cup_{r=1}^{p-1}B_{-1}(r)\Big),
$$
where $B_{-1}(r)=\bigl\{x: |x-r|_p \le p^{-1}\bigr\}$
(see formula (\ref{79}) in Proposition~\ref{pr1-int}).
The equation  (\ref{62.0-3**}) is an analog of the {\em refinement equation}
generating the Haar MRA in the real analysis.
Using this idea, the notion of $p$-adic MRA was introduced and a general
scheme for its construction was described in~\cite{S-Sk-1}.
The scheme was realized for construction $2$-adic Haar MRA
with using (\ref{62.0-3**-1}) as the {\em generating refinement equation}.
In contrast to the real setting, the {\em refinable function} $\phi$ generating
the Haar MRA is {\em periodic}, which {\em never holds} for real
refinable functions. Due to this fact, there exist {\em infinity many different}
orthonormal wavelet bases in the same Haar MRA (see~\cite{S-Sk-1}). One of them
coincides with Kozyrev's wavelet basis (\ref{62.0-1}).
From the standpoint of results of the papers~\cite{Kh-Sh-Sk},~\cite{Al-Ev-Sk},
in~\cite{S-Sk-1} all compactly supported wavelet Haar bases were constructed.

It turned out that the above-mentioned $p$-adic wavelets are eigenfunctions of
$p$-adic pseudo-differential operators~\cite{Al-Kh-Sh3}--\cite{Al-Kh-Sh5},
~\cite{Kh-Sh1},~\cite{Kh-Sh2},~\cite{Koz0},~\cite{Koz2} (see also Sec.~\ref{s4}).
Thus the spectral theory of $p$-adic pseudo-differential operators is related
to the wavelet theory. On the other hand, it is well-known that numerous models
connected with $p$-adic differential equations use pseudo-differential operators
(see~\cite{Kh2},~\cite{Koch3},~\cite{Vl-V-Z} and the references therein).
This is closely related to the fact that for the $p$-adic analysis associated with the
mapping $\bQ_p \to \bC$, the operation of differentiation is {\it not defined\/},
and as a result, many models connected with $p$-adic differential equations
use pseudo-differential operators, in particular, the fractional operator
$D^{\alpha}$  (see the above-mentioned papers and books).
These two facts imply that study of wavelets is important since it gives
a new powerful technique for solving $p$-adic problems.

{\bf Contents of the paper.}
The main goal of this paper is to construct a {\em countable family of
new compactly supported non-Haar $p$-adic wavelet bases} in ${\cL}^2(\bQ_p)$.
Another goal is to study the connections between
{\em wavelet analysis and spectral analysis of $p$-adic pseudo-differential operators}.
In addition, we use our results to solve the Cauchy problems for $p$-adic
pseudo-differential equations.

In Sec.~\ref{s2}, we recall some facts from the theory of $p$-adic
distributions~\cite{G-Gr-P}, \cite{Taib1}--~\cite{Vl-V-Z}. In particular,
in Subsec.~\ref{s2.2}, some facts from the theory of the $p$-adic Lizorkin
spaces of test functions $\Phi(\bQ_p^n)$ and distributions $\Phi'(\bQ_p^n)$
are recalled (for details, see~\cite{Al-Kh-Sh3}).

In Sec.~\ref{s3}, {\it non-Haar $p$-adic compactly supported
wavelet bases} are introduced.
In Subsec.~\ref{s3.1}, we construct the non-Haar basis
(\ref{62.1}) which was introduced in the preprint~\cite{Kh-Sh1}
(for the brief review see~\cite{Kh-Sh2}).
In contrast to (\ref{62.0-1}), for the basis (\ref{62.1}) the number
of generating wavelet functions is not minimal, for example,
for $p=2$ we have $2^{m-1}$ wavelet functions (instead
of one as it is for (\ref{62.0-1}) and for classical wavelet
bases in real analysis).
The basis (\ref{62.1}) is the {\em non-Haar} wavelet basis, since it cannot be
constructed in the framework of the $p$-adic Haar MRA (see~\cite{S-Sk-1} and Theorem~\ref{th2}).
According to Remark~\ref{rem1}, Kozyrev's wavelet basis (\ref{62.0-1}) is
a particular case of the basis (\ref{62.1}) for $m=1$.
According to the same remark, our non-Haar wavelet basis (\ref{62.1}) can be
obtained by using the algorithm developed by the Benedettos~\cite{Ben-Ben}.
However, using our approach, we obtained the {\em explicit formulas} (\ref{62.1})
for this basis. Moreover, our technique allows to produce new wavelet bases
(see in Subsec.~\ref{s3.2}).
In Subsec.~\ref{s3.2}, using the proof scheme of~\cite[Theorem~1]{S-Sk-1},
we construct {\em infinitely many new different non-Haar wavelet bases}
(\ref{109-11}), (\ref{101-11}), (\ref{108-11}) which are distinct from the
basis (\ref{62.1}). These new bases cannot be obtained in the framework of
the standard scheme of the MRA~\cite{S-Sk-1}.
Our bases given by formulas (\ref{109-11}),
(\ref{101-11}), (\ref{108-11}) {\it cannot be constructed} by Benedettos'
method~\cite{Ben-Ben}.
For example, it is easy to see that the Fourier transform of our
generating wavelet-functions $\psi_{s}^{(m)[1]}(x)$, $s\in J_{p;m}$ defined by
(\ref{101-11}), (\ref{108-11}) and all their shifts {\em are not characteristic
functions} (see Remark~\ref{rem2}).
In Subsec.~\ref{s3.3}, $n$-dimensional non-Haar wavelet bases (\ref{62.8})
and (\ref{62.8-1}) are introduced as $n$-direct products of the corresponding
one-dimensional non-Haar wavelet bases.
All above wavelets belong to the Lizorkin space of test functions $\Phi(\bQ_p^n)$.
In Subsec.~\ref{s3.4}, the characterizations of the spaces of Lizorkin test functions and
distributions in terms of wavelets are given (see Lemma~\ref{lem-w-1**} and
Proposition~\ref{pr-w-2**}), which are very useful for solution of $p$-adic
pseudo-differential equations. The assertions of the type of
Lemma~\ref{lem-w-1**} and Proposition~\ref{pr-w-2**} were stated for ultrametric
Lizorkin spaces in~\cite{Al-Koz}.

In Sec.~\ref{s4}, the spectral theory of one class of $p$-adic multidimensional
pseudo-differential operators (\ref{64.3}) (which were introduced in~\cite{Al-Kh-Sh3})
is studied.
In Subsec.~\ref{s4.1},~\ref{s4.2}, we recall some facts on this class of
pseudo-differential operators defined in the Lizorkin space ${\cD}'(\bQ_p^n)$.
Our operators (\ref{64.3}) include the fractional
operator~\cite[\S2]{Taib1},~\cite[III.4.]{Taib3} and the
pseudo-differential operators studied in~\cite{Koch3},~\cite{Z1},~\cite{Z2}.
The Lizorkin spaces are {\it invariant\/} under our
pseudo-differential operators.
In Subsec.~\ref{s4.3}, by Theorems~\ref{th4.1},~\ref{th4.1-1} the
criterion (\ref{64.1***}) for multidimensional $p$-adic
pseudo-differential operators (\ref{64.3}) to have multidimensional
wavelets (\ref{62.8}) and (\ref{62.8-1}) as eigenfunctions is derived.
In particular, the multidimensional wavelets (\ref{62.8}) and (\ref{62.8-1})
are eigenfunctions of the Taibleson fractional operator (see
Corollaries~\ref{cor5}--\ref{cor7}).

In Sec.~\ref{s5}, the results of Sec.~\ref{s3},~\ref{s4} are used to solve
the Cauchy problems for $p$-adic evolutionary pseudo-differential equations.
Note that the Cauchy problem (\ref{76-sl}) was solved in~\cite{Al-Kh-Sh5}
for a particular case. These results give significant advance in the theory
of $p$-adic pseudo-differential equations. Moreover, since many $p$-adic models
use pseudo-differential operators (in particular, fractional operator),
these results can be used in applications.

It easy to see that formulas (\ref{109-11}), (\ref{101-11}), (\ref{108-11})
do not give description of all {\em non-Haar wavelet bases}.

Due to the results of Sec.~\ref{s3}, there arise
two important problems: to construct an analog of MRA scheme and describe all
compactly supported {\em non-Haar wavelet bases}. It is necessary
to verify if all {\em non-Haar wavelet bases} are given by formulas (\ref{109-11}),
(\ref{101-11}), (\ref{108-11})?

Taking into account representation (\ref{62.0-5}), it is natural to suggest that in this
case we must use the {\em refinement type equation}:
$$
\phi(x)=\sum_{b}\phi\Big(\frac{1}{p^m}x-\frac{b}{p^m}\Big),
\quad x\in \bQ_p,
$$
instead of the {\em Haar refinement equation} (\ref{62.0-3**}),
where $b=0$ or $b=b_{r}p^{r}+b_{r+1}p^{r+1}+\cdots+b_{m-1}p^{m-1}$,
$r=0,1,\dots,m-1$, $0\le b_j\le p-1$, $b_r\ne 0$.
This equation reflects the geometric fact that
the unit disc $B_{0}=\{x: |x|_p \le 1\}$ is represented by a sum of $p^m$
mutually {\em disjoint} discs $B_{-m}(b)=\{x:|x-b|_p \le p^{-m}\}$.

\section{Preliminary results in $p$-adic analysis}
\label{s2}

\subsection{$p$-Adic functions and distributions.}\label{s2.1}
We shall systematically use the notations and results from~\cite{Vl-V-Z}.
Let $\bN$, $\bZ$, $\bC$ be the sets of positive integers, integers,
complex numbers, respectively.

Any $p$-adic number $x\in\bQ_P$, $x\ne 0$, is represented in the {\em canonical form}
\begin{equation}
\label{8.1}
x=p^{\gamma}(x_0 + x_1p + x_2p^2 + \cdots )
\end{equation}
where $\gamma=\gamma(x)\in \bZ$, \ $x_k=0,1,\dots,p-1$, $x_0\ne 0$,
$k=0,1,\dots$. The series is convergent in the $p$-adic norm $|\cdot|_p$,
and one has $|x|_p=p^{-\gamma}$.
The {\it fractional part} of a number $x\in \bQ_p$ (given by (\ref{8.1}))
is defined as follows
\begin{equation}
\label{8.2**}
\{x\}_p=\left\{
\begin{array}{lll}
0,\quad \text{if} \quad \gamma(x)\ge 0 \quad  \text{or} \quad x=0,&&  \\
p^{\gamma}(x_0+x_1p+x_2p^2+\cdots+x_{|\gamma|-1}p^{|\gamma|-1}),
\quad \text{if} \quad \gamma(x)<0. && \\
\end{array}
\right.
\end{equation}

The space $\bQ_p^n=\bQ_p\times\cdots\times\bQ_p$ consists of points
$x=(x_1,\dots,x_n)$, where $x_j \in \bQ_p$, $j=1,2\dots,n$, \ $n\ge 2$.
The $p$-adic norm on $\bQ_p^n$ is
\begin{equation}
\label{8}
|x|_p=\max_{1 \le j \le n}|x_j|_p, \quad x\in \bQ_p^n,
\end{equation}
where $|x_j|_p$, $x_j\in \bQ_p$, is defined by (\ref{1}), $j=1,\dots,n$.
Denote by $B_{\gamma}^n(a)=\{x: |x-a|_p \le p^{\gamma}\}$ the ball
of radius $p^{\gamma}$ with the center at a point $a=(a_1,\dots,a_n)\in \bQ_p^n$
and by $S_{\gamma}^n(a)=\{x: |x-a|_p = p^{\gamma}\}
=B_{\gamma}^n(a)\setminus B_{\gamma-1}^n(a)$ its boundary (sphere),
$\gamma \in \bZ$. For $a=0$ we set $B_{\gamma}^n(0)=B_{\gamma}^n$ and
$S_{\gamma}^n(0)=S_{\gamma}^n$. For the case $n=1$ we will omit the
upper index $n$.
Here
\begin{equation}
\label{9}
B_{\gamma}^n(a)=B_{\gamma}(a_1)\times\cdots\times B_{\gamma}(a_n),
\end{equation}
where $B_{\gamma}(a_j)=\{x_j: |x_j-a_j|_p \le p^{\gamma}\}$ is a
disc of radius $p^{\gamma}$ with the center at a point $a_j\in \bQ_p$,
$j=1,2\dots,n$. Any two balls in $\bQ_p^n$ either are
disjoint or one contains the other. Every point of the ball is its
center.

\begin{Proposition}
\label{pr1-int}
{\rm (~\cite[I.3, Examples 1,2.]{Vl-V-Z})}
The disc $B_{\gamma}$ is represented by the sum of $p^{\gamma-\gamma'}$
{\em disjoint} discs $B_{\gamma'}(a)$, $\gamma'<\gamma$:
\begin{equation}
\label{79.0}
B_{\gamma}=B_{\gamma'}\cup\cup_{a}B_{\gamma'}(a),
\end{equation}
where $a=0$ and
$a=a_{-r}p^{-r}+a_{-r+1}p^{-r+1}+\cdots+a_{-\gamma'-1}p^{-\gamma'-1}$
are the centers of the discs $B_{\gamma'}(a)$, \
$0\le a_j\le p-1$, $j=-r,-r+1,\dots,-\gamma'-1$, $a_{-r}\ne 0$, \,
$r=\gamma,\gamma-1,\gamma-2,\dots,\gamma'+1$.
In particular, the disc $B_{0}$ is represented by the sum of $p$
{\em disjoint} discs
\begin{equation}
\label{79}
B_{0}=B_{-1}\cup\cup_{r=1}^{p-1}B_{-1}(r),
\end{equation}
where $B_{-1}(r)=\{x\in S_{0}: x_0=r\}=r+p\bZ_p$, $r=1,\dots,p-1$;
$B_{-1}=\{|x|_p\le p^{-1}\}=p\bZ_p$; and
$S_{0}=\{|x|_p=1\}=\cup_{r=1}^{p-1}B_{-1}(r)$. Here all the discs
are disjoint.
\end{Proposition}

We call covering (\ref{79.0}), (\ref{79}) the {\it canonical covering} of the disc $B_{0}$.

A complex-valued function $f$ defined on $\bQ_p^n$ is called
{\it locally-constant} if for any $x\in \bQ_p^n$ there exists
an integer $l(x)\in \bZ$ such that
$$
f(x+y)=f(x), \quad y\in B_{l(x)}^n.
$$
Let ${\cE}(\bQ_p^n)$ and ${\cD}(\bQ_p^n)$ be the
linear spaces of locally-constant $\bC$-valued functions on $\bQ_p^n$
and locally-constant $\bC$-valued functions with compact supports
(so-called test functions), respectively; ${\cD}(\bQ_p)$,
${\cE}(\bQ_p)$~\cite[VI.1.,2.]{Vl-V-Z}.
If $\varphi \in {\cD}(\bQ_p^n)$, according to Lemma~1 from~\cite[VI.1.]{Vl-V-Z},
there exists $l\in \bZ$, such that
$$
\varphi(x+y)=\varphi(x), \quad y\in B_l^n, \quad x\in \bQ_p^n.
$$
The largest of the numbers $l=l(\varphi)$ is called the
{\it parameter of constancy} of the function $\varphi$.
Let us denote by ${\cD}^l_N(\bQ_p^n)$ the finite-dimensional space of
test functions from ${\cD}(\bQ_p^n)$ with supports in the ball $B_N^n$
and with parameters of constancy $\ge l$~\cite[VI.2.]{Vl-V-Z}.
We have ${\cD}^l_N(\bQ_p^n) \subset {\cD}^{l'}_{N'}(\bQ_p^n)$, \
$N\le N'$, \ $l\ge l'$.

\begin{Lemma}
\label{lem-four-1}
{\rm(~\cite[VI.5.,(5.2')]{Vl-V-Z})}
Any function $\varphi \in {\cD}^l_N(\bQ_p^n)$ can be represented as a {finite}
linear combination
$$
\varphi(x)=\sum_{\nu=1}^{p^{n(N-l)}}\varphi(a^{\nu})\Delta_{l}(x-a^{\nu}),
\quad x\in \bQ_p^n,
$$
where $\Delta_{l}(x-a^{\nu})=\Omega(p^{-l}|x-a^{\nu}|_p)](x)$ is the characteristic
function of the ball $B_{l}^n(a^{\nu})$, and the points
$a^{\nu}=(a_1^{\nu},\dots a_n^{\nu})\in B_N^n$ do not depend on $\varphi$
and are such that the bolls
$B_{l}^n(a^{\nu})$, $\nu=1,\dots,p^{n(N-l)}$, are disjoint and cover
the ball $B_{N}^n$.
\end{Lemma}

Denote by ${\cD}'(\bQ_p^n)$ the set of all linear functionals on
${\cD}(\bQ_p^n)$~\cite[VI.3.]{Vl-V-Z}.

The Fourier transform of $\varphi\in {\cD}(\bQ_p^n)$ is defined by the
formula
$$
F[\varphi](\xi)=\int_{\bQ_p^n}\chi_p(\xi\cdot x)\varphi(x)\,d^nx,
\quad \xi \in \bQ_p^n,
$$
where $d^n x=dx_1\cdots dx_n$ is the Haar measure such that
$\int_{|\xi|_p\le 1}\,d^nx=1$;
$\chi_p(\xi\cdot x)=\chi_p(\xi_1x_1)\cdots\chi_p(\xi_nx_n)$;
$\xi\cdot x$ is the scalar product of vectors and
$\chi_p(\xi_jx_j)=e^{2\pi i\{\xi_jx_j\}_p}$
are additive characters, $\{x\}_p$ is the {\it fractional part} (\ref{8.2**})
of a number $x\in \bQ_p$.

\begin{Lemma}
\label{lem-four-2}
{\rm(~\cite[Lemma~A.]{Taib1},~\cite[III,(3.2)]{Taib3},~\cite[VII.2.]{Vl-V-Z})}
The Fourier transform is a linear isomorphism ${\cD}(\bQ_p^n)$ into
${\cD}(\bQ_p^n)$. Moreover,
\begin{equation}
\label{12}
\varphi(x) \in {\cD}^l_N(\bQ_p^n) \quad \text{iff} \quad
F\big[\varphi(x)\big](\xi) \in {\cD}^{-N}_{-l}(\bQ_p^n).
\end{equation}
\end{Lemma}

The Fourier transform $F[f]$ of a distribution
$f\in {\cD}'(\bQ_p^n)$ is defined by the relation
$\langle F[f],\varphi\rangle=\langle f,F[\varphi]\rangle$,
for all $\varphi\in {\cD}(\bQ_p^n)$.

Let $A$ be a matrix and $b\in \bQ_p^n$. Then for a distribution
$f\in{\cD}'(\bQ_p^n)$ the following relation holds~\cite[VII,(3.3)]{Vl-V-Z}:
\begin{equation}
\label{14}
F[f(Ax+b)](\xi)
=|\det{A}|_p^{-1}\chi_p\big(-A^{-1}b\cdot \xi\big)F[f(x)]\big(A^{-1}\xi\big),
\end{equation}
where $\det{A} \ne 0$.
According to~\cite[IV,(3.1)]{Vl-V-Z},
\begin{equation}
\label{14.1}
F[\Omega(p^{-k}|\cdot|_p)](x)=p^{nk}\Omega(p^k|x|_p), \quad k\in \bZ, \qquad x \in \bQ_p^n.
\end{equation}
In particular, $F[\Omega(|\xi|_p)](x)=\Omega(|x|_p)$. Here
$\Omega(t)$ is the characteristic function of the segment $[0,1]\subset\bR$,

\subsection{$p$-Adic Lizorkin spaces}\label{s2.2}
According to~\cite{Al-Kh-Sh3},~\cite{Al-Kh-Sh4}, the $p$-adic
{\it Lizorkin space of test functions\/} is defined as
$$
\Phi(\bQ_p^n)=\{\phi: \phi=F[\psi], \, \psi\in \Psi(\bQ_p^n)\},
$$
where $\Psi(\bQ_p^n)
=\{\psi(\xi)\in \cD(\bQ_p^n): \psi(0)=0\}$.
The space $\Phi(\bQ_p^n)$ can be equipped with the
topology of the space $\cD(\bQ_p^n)$ which makes it a
complete space. In view of Lemma~\ref{lem-four-2}, the Lizorkin space
$\Phi(\bQ_p^n)$ admits the following characterization.

\begin{Lemma}
\label{lem1}
{\rm (~\cite{Al-Kh-Sh3},~\cite{Al-Kh-Sh4})}
{\rm (a)} $\phi\in \Phi(\bQ_p^n)$ iff $\phi\in \cD(\bQ_p^n)$ and
\begin{equation}
\label{54}
\int_{\bQ_p^n}\phi(x)\,d^nx=0.
\end{equation}

{\rm (b)} $\phi \in {\cD}^l_N(\bQ_p^n)\cap\Phi(\bQ_p^n)$, i.e.,
$\int_{B^n_{N}}\phi(x)\,d^nx=0$,
iff \ $\psi=F^{-1}[\phi]\in {\cD}^{-N}_{-l}(\bQ_p^n)\cap\Psi(\bQ_p^n)$,
i.e., $\psi(\xi)=0$, $\xi \in B^n_{-N}$.
\end{Lemma}

Suppose that $\Phi'(\bQ_p^n)$ and $\Psi'(\bQ_p^n)$ denote the topological
dual of the spaces $\Phi(\bQ_p^n)$ and $\Psi(\bQ_p^n)$, respectively.
We call $\Phi'(\bQ_p^n)$ the space of  $p$-adic {\em Lizorkin
distributions}.
The space $\Phi'(\bQ_p^n)$ can be obtained from $\cD'(\bQ_p^n)$ by
``sifting out'' constants. Thus two distributions in $\cD'(\bQ_p^n)$
differing by a constant are indistinguishable as elements of $\Phi'(\bQ_p^n)$.

We define the Fourier transform of
$f\in \Phi'(\bQ_p^n)$ and $g\in \Psi'(\bQ_p^n)$
respectively by formulas $\langle F[f],\psi\rangle=\langle f,F[\psi]\rangle$,
for all $\psi\in \Psi(\bQ_p^n)$, and
$\langle F[g],\phi\rangle=\langle g,F[\phi]\rangle$, for all
$\phi\in \Phi(\bQ_p^n)$.
It is clear that $F[\Phi'(\bQ_p^n)]=\Psi'(\bQ_p^n)$
and $F[\Psi'(\bQ_p^n)]=\Phi'(\bQ_p^n)$~\cite{Al-Kh-Sh3}.

Recall that in the real setting the Lizorkin spaces were introduced in the
excellent papers by P.~I.~Lizorkin~\cite{Liz1},~\cite{Liz3}.

\section{Non-Haar $p$-adic wavelet bases}
\label{s3}

\subsection{One non-Haar wavelet basis in ${\cL}^2(\bQ_p)$.}\label{s3.1}
It is well known that
$\bQ_p=B_{0}(0)\cup\cup_{\gamma=1}^{\infty}S_{\gamma}$, where
$S_{\gamma}=\{x\in \bQ_p: |x|_p = p^{\gamma}\}$. Due to
(\ref{8.1}), $x\in S_{\gamma}$, $\gamma\ge 1$, if and only if
$x=x_{-\gamma}p^{-\gamma}+x_{-\gamma+1}p^{-\gamma+1}+\cdots+x_{-1}p^{-1}+\xi$,
where $x_{-\gamma}\ne 0$, $\xi \in B_{0}(0)$. Since
$x_{-\gamma}p^{-\gamma}+x_{-\gamma+1}p^{-\gamma+1}
+\cdots+x_{-1}p^{-1}\in I_p$, we have a ``natural'' decomposition of
$\bQ_p$ into a union of mutually  disjoint discs:
$$
\bQ_p=\bigcup\limits_{a\in I_p}B_{0}(a).
$$
Therefor,
$$
I_p=\{a=p^{-\gamma}\big(a_{0}+a_{1}p+\cdots+a_{\gamma-1}p^{\gamma-1}\big):
\qquad\qquad\qquad\qquad
$$
\begin{equation}
\label{62.0**}
\gamma\in \bN; \, a_j=0,1,\dots,p-1; \, j=0,1,\dots,\gamma-1\}
\end{equation}
is a {\em ``natural'' set of shifts} for $\bQ_p$, which will
be used in the sequel.
This set $I_p$ can be identified with the factor group $\bQ_p/\bZ_p$.

Let
$$
J_{p;m}=\{s=p^{-m}\big(s_{0}+s_{1}p+\cdots+s_{m-1}p^{m-1}\big):
\qquad\qquad\qquad\qquad
$$
\begin{equation}
\label{62.0*}
\qquad
s_j=0,1,\dots,p-1; \, j=0,1,\dots,m-1; s_0\ne 0\},
\end{equation}
where $m\ge 1$ is a {\it fixed} positive integer.

Let us introduce the set of $(p-1)p^{m-1}$ functions
\begin{equation}
\label{62.1-11}
\theta_{s}^{(m)}(x)=\chi_p(sx)\Omega\big(|x|_p\big), \quad s\in J_{p;m},
\quad x\in \bQ_p,
\end{equation}
and the family of functions generated by their dilatations and shifts:
\begin{equation}
\label{62.1}
\theta_{s;\,j a}^{(m)}(x)=p^{-j/2}\chi_p\big(s(p^{j}x-a)\big)
\Omega\big(|p^{j}x-a|_p\big), \quad x\in \bQ_p,
\end{equation}
where $s\in J_{p;m}$, $j\in \bZ$, $a\in I_p$, $\Omega(t)$ is the
characteristic function of the segment $[0,1]\subset\bR$.

\begin{Theorem}
\label{th2}
The functions {\rm(\ref{62.1})} form an orthonormal {\em non-Haar $p$-adic wavelet
basis} in ${\cL}^2(\bQ_p)$.
\end{Theorem}

\begin{proof}
1. Consider the scalar product
$$
\big(\theta_{s';\,j'  a'}^{(m)}(x),\theta_{s;\,j a}^{(m)}(x)\big)
=p^{-(j+j')/2}\int_{\bQ_p}\chi_p\big(s'(p^{j'}x-a')-s(p^{\gamma}x-a)\big)
$$
\begin{equation}
\label{62.3}
\quad
\times
\Omega\big(|p^{j}x-a|_p\big)\Omega\big(|p^{j'}x-a'|_p\big)\,dx.
\end{equation}
If $j \le j'$, according to formula~\cite[VII.1]{Vl-V-Z},~\cite{Koz0}
\begin{equation}
\label{62.4-1}
\Omega\big(|p^{j}x-a|_p\big)\Omega\big(|p^{j'}x-a'|_p\big)
=\Omega\big(|p^{j}x-a|_p\big)\Omega\big(|p^{j'-j}a-a'|_p\big),
\end{equation}
(\ref{62.3}) can be rewritten as
$$
\big(\theta_{s';\, j' a'}^{(m)}(x),\theta_{s;\,j a}^{(m)}(x)\big)
=p^{-(j+j')/2}\Omega\big(|p^{j'-j}a-a'|_p\big)
\qquad\qquad\qquad\qquad
$$
\begin{equation}
\label{62.4}
\times
\int_{\bQ_p}
\chi_p\big(s'(p^{j'}x-a')-s(p^{j}x-a)\big)\Omega\big(|p^{j}x-a|_p\big)\,dx.
\end{equation}

Let $j < j'$. Making the change of variables $\xi=p^{j}x-a$
and taking into account (\ref{14.1}), we obtain from (\ref{62.4})
$$
\big(\theta_{s';\, j' a'}^{(m)}(x),\theta_{s;\,j a}^{(m)}(x)\big)
=p^{-(j+j')/2}\chi_p\big(s'(p^{j'-j}a-a')\big)
\qquad\qquad\qquad\qquad
$$
$$
\qquad\qquad
\times
\Omega\big(|p^{j'-j}a-a'|_p\big)
\int_{\bQ_p}\chi_p\big((p^{j'-j}s'-s)\xi\big)
\Omega\big(|\xi|_p\big)\,d\xi
$$
$$
=p^{-(j+j')/2}\chi_p\big(s'(p^{j'-j}a-a')\big)
\quad
$$
\begin{equation}
\label{62.5}
\times
\Omega\big(|p^{j'-j}a-a'|_p\big)\Omega\big(|p^{j'-j}s'-s|_p\big).
\end{equation}
Since
$$
\begin{array}{rclrcl}
\displaystyle
p^{j'-j}s'&=&p^{j'-j-m}
\big(s_{0}'+s_{1}'p+\cdots+s_{j-1}'p^{m-1}\big), \smallskip \\
s&=&p^{-m}\big(s_{0}+s_{1}p+\cdots+s_{j-1}p^{m-1}\big), \\
\end{array}
$$
where $s_{0}',s_{0}\ne 0$, $j'-j\le 1$, it is clear that the
fractional part $\{p^{j'-j}s'-s\}_p\ne 0$. Thus
$\Omega\big(|p^{j'-j}s'-s|_p\big)=0$ and
$\big(\theta_{s';\, j' a'}^{(m)}(x),\theta_{s;\, j a}^{(m)}(x)\big)=0$.

Consequently, the scalar product $\big(\theta_{s';\, j' a'}^{(m)}(x),
\theta_{s;\,j a}^{(m)}(x)\big)=0$ can be nonzero
only if $j=j'$. In this case (\ref{62.5}) implies
$$
\big(\theta_{s';\, j a'}^{(m)}(x),\theta_{s;\,j a}^{(m)}(x)\big)
\qquad\qquad\qquad\qquad\qquad\qquad\qquad\qquad
$$
\begin{equation}
\label{62.6}
=p^{-j}
\chi_p\big(s'(a-a')\big)\Omega\big(|a-a'|_p\big)\Omega\big(|s'-s|_p\big),
\end{equation}
where $\Omega\big(|a-a'|_p\big)=\delta_{a'a}$,
$\Omega\big(|s'-s|_p\big)=\delta_{s's}$, and $\delta_{s's}$,
$\delta_{a'a}$ are the Kronecker symbols.

Since $\int_{\bQ_p}\Omega\big(|p^{j}x-a|_p\big)\,dx=p^{j}$, formulas
(\ref{62.5}), (\ref{62.6}) imply that
\begin{equation}
\label{62.7}
\big(\theta_{s';\, j' a'}^{(m)}(x),\theta_{s;\,j a}^{(m)}(x)\big)
=\delta_{s's}\delta_{j'j}\delta_{a'a}.
\end{equation}
Thus the system of functions (\ref{62.1}) is orthonormal.

To prove the completeness of the system of functions (\ref{62.1}),
we repeat the corresponding proof~\cite{Koz0} almost word for word.
Recall that the system of the characteristic functions of the discs
$B_{k}(0)$ is complete in ${\cL}^2(\bQ_p)$. Consequently, taking into
account that the system of functions
$\{\theta_{s;\,j a}^{(m)}(x): s\in J_{p;m}; j\in \bZ, a\in I_p\}$
is invariant under dilatations and shifts, in order to prove that
it is a complete system, it is sufficient to verify the Parseval identity
for the characteristic function $\Omega(|x|_p)$.

If $0\le j$, according to (\ref{62.4-1}), (\ref{14.1}),
$$
\big(\Omega(|x|_p),\theta_{s;\,j a}^{(m)}(x)\big)
=p^{-j/2}\Omega\big(|-a|_p\big)
\int_{\bQ_p}\chi_p\big(s(p^{j}x-a)\big)\Omega\big(|x|_p\big)\,dx
$$
$$
\qquad\qquad
=p^{-j/2}\chi_p\big(-sa)\big)
\Omega\big(|sp^{j}|_p\big)\Omega\big(|-a|_p\big)
$$
\begin{equation}
\label{62.7-11}
\qquad\quad
=\left\{
\begin{array}{lll}
0, && a\ne 0, \\
0, && a=0, \quad j \le m-1, \\
p^{-j/2}, && a=0, \quad j \ge m. \\
\end{array}
\right.
\end{equation}
If $0> j$, according to (\ref{62.4-1}), (\ref{14.1}),
$$
\big(\Omega(|x|_p),\theta_{s;\,j a}^{(m)}(x)\big)
=p^{-j/2}\Omega\big(|p^{-j}a|_p\big)
\int_{\bQ_p}\chi_p\big(s(p^{j}x-a)\big)
\Omega\big(|p^{j}x-a|_p\big)\,dx
$$
\begin{equation}
\label{62.7-12}
=p^{j/2}\Omega\big(|p^{-j}a|_p\big)
\int_{\bQ_p}\chi_p\big(s\xi\big)\Omega\big(|\xi|_p\big)\,d\xi
=p^{j/2}\Omega\big(|p^{-j}a|_p\big)\Omega\big(|s|_p\big)=0.
\end{equation}
Thus,
$$
\sum_{s\in J_{p;m};j\in \bZ,a\in I_p}
\big|\big(\Omega(|x|_p),\theta_{s;\,j a}^{(m)}(x)\big)\big|^2
=\sum_{j=m}^{\infty}\sum_{s\in J_{p;m}}p^{-j}
\qquad\qquad\qquad\quad
$$
$$
\qquad\qquad
=p^{m-1}(p-1)\frac{p^{-m}}{1-p^{-1}}=1
=\big|\big(\Omega(|x|_p),\Omega(|x|_p)\big|^2.
$$

Thus the system of functions (\ref{62.1}) is an orthonormal basis in
${\cL}^2(\bQ_p)$.

2. Since elements of basis (\ref{62.1}) can be obtained
by dilatations and shifts of the set of $(p-1)p^{m-1}$ functions
(\ref{62.1-11}), it is the $p$-adic wavelet basis.

3. According to~\cite{S-Sk-1}, the Haar wavelet functions are constructed by
the Haar type {\em refinement equation} (\ref{62.0-3**}). In particular, it is easy
to see that Kozyrev's wavelet functions (\ref{62.0-1-0}) can be expressed in terms
of the {\em refinable function} (\ref{62.0-3**-1}) as
\begin{equation}
\label{62.0-1-2}
\theta_{k}(x)=\chi_p(p^{-1}kx)\Omega\big(|x|_p\big)
=\sum_{r=0}^{p-1}h_{k r}\phi\Big(\frac{1}{p}x-\frac{r}{p}\Big),
\quad x\in \bQ_p,
\end{equation}
where $h_{k r}=e^{2\pi i\{\frac{kr}{p}\}_p}$, $r=0,1,\dots,p-1$, \
$k=1,2,\dots,p-1$.
In this case the wavelet functions $\theta_{k}(x)=\chi_p(p^{-1}kx)\Omega\big(|x|_p\big)$
takes values in the set $\{e^{2\pi i\frac{kr}{p}}:r=0,1,\dots,p-1\}$ of $p$ elements
on the discs $B_{-1}(r)$, $r=0,1,\dots,p-1$.
Thus any wavelet function $\theta_{k}(x)$ is represented as a linear combination
of the characteristic functions of the disks of the radius of $p^{-1}$.

In contrast to the Kozyrev wavelet basis (\ref{62.0-1}), the number
of generating wavelet functions (\ref{62.1-11}) for the wavelet basis (\ref{62.1})
is not minimal. For example, if $p=2$, then we have $2^{m-1}$
wavelet functions (instead of one as it is for (\ref{62.0-1}) and for classical
wavelet bases in real analysis.

Let $B_{0}=\cup_{b}B_{-m}(a)\cup B_{-m}$ be the {\it canonical covering}
(\ref{79.0}) of the disc $B_{0}$ with $p^{m}$ discs, $m\ge 1$, where
$b=0$ and $b=b_{r}p^{r}+b_{r+1}p^{r+1}+\cdots+b_{m-1}p^{m-1}$
is the center of the discs $B_{-m}$ and $B_{-m}(b)$, respectively,
$0\le b_j\le p-1$, $j=r,r+1,\dots,m-1$, $b_r\ne 0$, \ $r=0,1,2,\dots,m-1$.

For $x\in B_{-m}(b)$, $s\in J_{p;m}$, we have
$x=b+p^{m}\big(y_{0}+y_{1}p+y_{2}p^2+\cdots\big)$,
$s=p^{-m}\big(s_{0}+s_{1}p+\cdots+s_{m-1}p^{m-1}\big)$, $s_0\ne 0$;
$sx=sb+\xi$, $\xi\in \bZ_p$;
and $\{sx\}_p=\{sb\}_p=\{p^{r-m}\big(b_{r}+a_{r+1}p+\cdots+b_{m-1}p^{m-r-1}\big)
\big(s_{0}+s_{1}p+\cdots+s_{m-1}p^{m-1}\big)\}_p$, \
$r=0,1,2,\dots,m-1$, (see (\ref{8.2**})).
Thus,
$$
\theta_{s}^{(m)}(x)=\chi_p(sx)\Omega\big(|x|_p\big)
\qquad\qquad\qquad\qquad\qquad\qquad\qquad\qquad
$$
\begin{equation}
\label{62.0-2-1*}
\qquad
=\left\{
\begin{array}{lll}
0, && |x|_p\ge p, \\
e^{2\pi i\{sa\}_p}, && x\in B_{-m}(b), \quad
b=\sum_{l=r}^{m-1}b_{l}p^{l}, \\
1, && x\in B_{-m}, \\
\end{array}
\right.
\end{equation}
where $0\le b_j\le p-1$, $j=r,\dots,m-1$, $b_r\ne 0$, $r=0,1,\dots,m-1$;
\, $s=p^{-m}\big(s_{0}+s_{1}p+\cdots+s_{m-1}p^{m-1}\big)$,
$0\le s_j\le p-1$, $j=0,1,\dots,m-1$, $s_0\ne 0$.
Thus the function $\theta_{s}^{(m)}(x)=\chi_p(sx)\Omega\big(|x|_p\big)$
takes values in the set $\{e^{2\pi i\{sb\}_p}:b=\sum_{l=r}^{m-1}b_{l}p^{l},
0\le b_j\le p-1, j=r,\dots,m-1, b_r\ne 0, r=0,1,\dots,m-1\}$
of $p^m$ elements on the discs $B_{-m}(b)$. By using (\ref{62.0-2-1*}),
one can see that in contrast to the Kozyrev wavelet functions (\ref{62.0-1-0}),
any wavelet function $\theta_{s}^{(m)}(x)$ is represented as a
linear combination of the characteristic functions of the disks of the radius
of $p^{-m}$:
\begin{equation}
\label{62.0-5}
\theta_{s}^{(m)}(x)=\chi_p(sx)\Omega\big(|x|_p\big)
=\sum_{b}{\tilde h}_{s b}\phi\Big(\frac{1}{p^m}x-\frac{b}{p^m}\Big),
\end{equation}
$x\in \bQ_p$, where ${\tilde h}_{s b}=e^{2\pi i\{sb\}_p}$,
$b=0$ or $b=b_{r}p^{r}+b_{r+1}p^{r+1}+\cdots+b_{m-1}p^{m-1}$,
$r=0,1,\dots,m-1$, $0\le b_j\le p-1$, $b_r\ne 0$; $s\in J_{p;m}$.

According to formulas (\ref{62.0-2-1*}), (\ref{62.0-5}), the wavelet
function (\ref{62.1-11}) {\em cannot be represented} in terms of the
{\em Haar refinable function} $\phi(x)=\Omega\big(|x|_p\big)$ (which is
a solution of the Haar {\it refinement equation} (\ref{62.0-3**})), i.e.,
in the form
$$
\theta_{s}^{(m)}(x)=\sum_{a\in I_p}\beta_a\phi(p^{-1}x-a),
\quad \beta_a\in \bC.
$$
Consequently, the wavelet basis (\ref{62.1}) is non-Haar type for $m\ge 2$.
\end{proof}

Making the change of variables $\xi=p^{j}x-a$ and taking into
account (\ref{14.1}), we obtain
$\int{\bQ_p}\theta_{s;\,j a}^{(m)}(x)\,dx
=p^{j/2}\int\limits_{\bQ_p}\chi_p\big(s\xi\big)\Omega\big(|\xi|_p\big)\,d\xi
=p^{j/2}\Omega\big(|s|_p\big)=0$,
i.e., according to Lemma~\ref{lem1}, the wavelet function
$\theta_{s;\,j a}^{(m)}(x)$ belongs to the Lizorkin space $\Phi(\bQ_p)$.

\begin{Remark}
\label{rem1} \rm
1. Kozyrev's wavelet basis (\ref{62.0-1}) is a particular case of the wavelet
basis (\ref{62.1}) for $m=1$. Indeed, in this case we have
$\theta_{s;\,j a}^{(1)}(x)\equiv \theta_{k;\,j a}(x)$,
where $s\equiv p^{-1}k$, $k=1,2,\dots,p-1$; \,$j\in \bZ$, $a\in I_p$.

2. Our non-Haar wavelet basis (\ref{62.1}) can be constructed in the framework of the
approach developed by J.~J.~Benedetto and R.~L.~Benedetto~\cite{Ben-Ben}
\footnote{Here we follow the referee report of the previous version of our paper.}.
In the notation~\cite{Ben-Ben}, $H^{\perp}$ is the ball $B_0$, and $A^*_1:\bQ_p\to\bQ_p$
is multiplication by $p^{-1}$, and $W=(A^*_1)^{m-1}H^{\perp}=B_{m-1}$.
In addition, the ``choice coset of representatives'' $\cD$ in~\cite{Ben-Ben}
is the set (\ref{62.0**}) of shifts $I_p$, and the set $J_{p;m}$ given by (\ref{62.0*})
is precisely the set $\cD\cap\big((A^*_1W)\setminus W\big)$ that appears in
equation~\cite[(4.1)]{Ben-Ben}.
The set $J_{p;m}$ consists of $N=p^{m}-p^{m-1}$ elements, where $N$ is the number
of wavelet generators.

The algorithm of~\cite{Ben-Ben} starts with $N$ sets $\Omega{s,0}$ and $N$
local translation functions $T_s$, one for each $s\in J_{p;m}$. In order to
construct the wavelets (\ref{62.1}) by using the algorithm of~\cite{Ben-Ben},
we set
$$
\Omega_{s,0}=B_0\big(p^{-m}(s_{1}p+\cdots+s_{m-1}p^{m-1})\big)
\qquad\qquad\quad
$$
$$
\qquad
=p^{-m}(s_{1}p+\cdots+s_{m-1}p^{m-1})+B_0\subset B_{m-1},
$$
i.e., remove the term $s_0$ and add the set $B_0=H^{\perp}$. We also define
$$
T_s:B_{m-1}\to B_{m}\setminus B_{m-1} \quad \text{by} \quad T_s(w)=w+s_0p^{-m},
$$
so that $T_s$ maps $\Omega_{s,0}$ to $B_0(s)$ by translation. It is easy to verify
that this data fit the requirements of~\cite{Ben-Ben}: each $\Omega_{s,0}$ is
$(\tau,\cD)$-congruent to $H^{\perp}=B_0$, the union of all such sets contains
a neighborhood of $0$, each $T_s$ has the form required by formula~\cite[(4.1)]{Ben-Ben},
and for each $s\ne s'$ in $J_{p;m}$, one of the two compatibility
condition~\cite[(4.2) or (4.3) ]{Ben-Ben} holds.
Thus, by using the algorithm of~\cite{Ben-Ben}, one can produce the wavelets (\ref{62.1}).
Moreover, according to (\ref{o-64.7}), any wavelet $\theta_{s}^{(m)}$ is the Fourier transform
of the characteristic function of each disc $B_0(s)$, $s\in J_{p;m}$. Recall that
the algorithm of~\cite{Ben-Ben} only allows the construction of wavelet functions whose
Fourier transforms are the characteristic functions
of some sets (see~\cite[Proposition~5.1.]{Ben-Ben}).
\end{Remark}

\subsection{Countable family of non-Haar wavelet bases in ${\cL}^2(\bQ_p)$.}\label{s3.2}
Now, using the proof scheme of~\cite[Theorem~1]{S-Sk-1}, we construct
infinitely many different non-Haar wavelet bases, which are distinct from the basis (\ref{62.1}).

In what follows, we shall write the $p$-adic number
$a=p^{-\gamma}\big(a_{0}+a_{1}p+\cdots+a_{\gamma-1}p^{\gamma-1}\big)\in I_p$, \
$a_{j}=0,1,\dots,p-1$, \ $j=0,1,\dots,\gamma-1$, in the form $a=\frac{r}{p^{\gamma}}$,
where $r=a_{0}+a_{1}p+\cdots+a_{\gamma-1}p^{\gamma-1}$.

Since the $p$-adic norm is non-Archimedean, it is easy to see that the wavelet
functions (\ref{62.1-11}) have the following  property:
\begin{equation}
\label{100-11}
\theta_{s}^{(m)}(x\pm 1)=\chi_p(\pm s)\theta_{s}^{(m)}(x), \quad s\in J_{p;m}.
\end{equation}

\begin{Theorem}
\label{th4-11}
Let $\nu=1,2,\dots$. The functions
\begin{equation}
\label{101-11}
\psi_{s}^{(m)[\nu]}(x)=\sum_{k=0}^{p^\nu-1}\alpha_{s;k}\theta_{s}^{(m)}\Big(x-\frac{k}{p^\nu}\Big),
\qquad s\in J_{p;m},
\end{equation}
are wavelet functions if and only if
\begin{equation}
\label{108-11}
\alpha_{s;k}=p^{-\nu}\sum_{r=0}^{p^\nu-1}\gamma_{s;r}e^{-2i\pi\frac{-s+r}{p}k},
\end{equation}
where $\gamma_{s;k}\in \bC$, $|\gamma_{s;k}|=1$, $k=0,1,\dots,p^\nu-1$, \,$s\in J_{p;m}$.
\end{Theorem}

\begin{proof}
Let $\psi_{s}^{(m)[\nu]}(x)$ be defined by (\ref{101-11}), $s\in J_{p;m}$.
According to Theorem~\ref{th2}, $\{\theta_{s}^{(m)}(\cdot-a), s\in J_{p;m}, a\in I_p\}$
is an orthonormal system. Hence, taking into account (\ref{100-11}),
we see that $\psi_{s}^{(m)[\nu]}$ is orthogonal to $\psi_{s}^{[\nu]}(x)(\cdot-a)$
whenever $a\in I_p$, $a\ne \frac{k}{p^\nu}$, $k=0,1,\dots p^\nu-1$; \ $\nu=1,2,\dots$.
Thus the system  $\{\psi_{s}^{(m)[\nu]}(x-a), s\in J_{p;m}, a\in I_p\}$ is
orthonormal if and only if the system  consisting of the functions
$$
\psi_{s}^{(m)[\nu]}\Big(x-\frac{r}{p^\nu}\Big)
=\chi_p(-s)\alpha_{s;p^\nu-r}\theta_{s}^{(m)}(x)
+\chi_p(-s)\alpha_{s;p^\nu-r+1}\theta_{s}^{(m)}\Big(x-\frac{1}{p^\nu}\Big)
$$
$$
\qquad\qquad\qquad\qquad
+\cdots
+\chi_p(-s)\alpha_{s;p^\nu-1}\theta_{s}^{(m)}\Big(x-\frac{r-1}{p^\nu}\Big)
$$
$$
\qquad\qquad
+\alpha_{s;0}\theta_{s}^{(m)}\Big(x-\frac{r}{p^\nu}\Big)
+\alpha_{s;1}\theta_{s}^{(m)}\Big(x-\frac{r+1}{p^\nu}\Big)
$$
\begin{equation}
\label{104-11}
\qquad\qquad
+\cdots
+\alpha_{s;p^\nu-1-r}\theta_{s}^{(m)}\Big(x-\frac{p^\nu-1}{p^\nu}\Big),
\end{equation}
$r=0,\dots,p^\nu-1$, $s\in J_{p;m}$, is orthonormal, $\nu=1,2,\dots$.
Set
$$
\Xi_{\nu}^{[0]}=\left(\theta_{s}^{(m)},
\theta_{s}^{(m)}\left(\cdot-\frac{1}{p^\nu}\right),\dots,
\theta_{s}^{(m)}\left(\cdot-\frac{p^\nu-1}{p^\nu}\right)\right)^T,
$$
$$
\Xi_{\nu}^{[\nu]}=\left(\psi_{s}^{(m)[\nu]},
\psi_{s}^{(m)[\nu]}\left(\cdot-\frac{1}{p^\nu}\right),\dots,
\psi_{s}^{(m)[\nu]}\left(\cdot-\frac{p^\nu-1}{p^\nu}\right)\right)^T,
$$
where $T$ is the transposition operation.
By (\ref{104-11}), we have $\Xi_{s}^{[\nu]}=D_{s}\Xi_{s}^{[0]}$,
where
$$
D_{s}=
\qquad\qquad\qquad\quad\qquad\qquad\qquad\qquad\quad\qquad\qquad
\qquad\qquad\qquad\qquad\qquad
$$
\begin{equation}
\label{105-11}
\left(
\begin{array}{ccccc}
\alpha_{s;0}&\alpha_{s;1}&\ldots&\alpha_{s;p^\nu-2}&\alpha_{s;p^\nu-1} \\
\chi_p(-s)\alpha_{s;p^\nu-1}&\alpha_{s;0}&\ldots&\alpha_{s;p^\nu-3}&\alpha_{s;p^\nu-2} \\
\chi_p(-s)\alpha_{s;p^\nu-2}&\chi_p(-s)\alpha_{s;p^\nu-1}&\ldots&\alpha_{s;p^\nu-4}&\alpha_{s;p^\nu-3} \\
\hdotsfor{5} \\
\chi_p(-s)\alpha_{s;2}&\chi_p(-s)\alpha_{s;3}&\ldots&\alpha_{s;0}&\alpha_{s;1} \\
\chi_p(-s)\alpha_{s;1}&\chi_p(-s)\alpha_{s;2}&\ldots&\chi_p(-s)\alpha_{s;p^\nu-1}&\alpha_{s;0} \\
\end{array}
\right),
\end{equation}
and $s\in J_{p;m}$; \,$\nu=1,2,\dots$.
Due to orthonormality of  $\{\psi_{s}^{(m)[\nu]}(x)(\cdot-a), s\in J_{p;m}, a\in I_p\}$,
the coordinates of  $\Xi_{s}^{[\nu]}$ form an orthonormal system
if and only if the matrixes $D_{s}$ are unitary, $s\in J_{p;m}$.

Let $u_{s}=(\alpha_{s;0},\alpha_{s;1},\dots,\alpha_{s;p^\nu-1})^{T}$ be a vector
and let
$$
A_{s}=\left(
\begin{array}{cccccc}
0&0&\ldots&0&0&\chi_p(-s) \\
1&0&\ldots&0&0&0 \\
0&1&\ldots&0&0&0 \\
\hdotsfor{6} \\
0&0&\ldots&1&0&0 \\
0&0&\ldots&0&1&0 \\
\end{array}
\right)
$$
be a $p^\nu\times p^\nu$ matrix, $s\in J_{p;m}$.
It is not difficult to see that
$$
A_{s}^ru_{\nu}
=\big(\chi_p(-s)\alpha_{s;p^\nu-r},\chi_p(-s)\alpha_{s;p^\nu-r+1},\dots,
\chi_p(-s)\alpha_{s;p^\nu-1},
$$
$$
\qquad\quad\qquad\qquad
\alpha_{s;0},\alpha_{s;1},\dots,\alpha_{s;p^\nu-r-1}\big)^{T},
$$
where $r=0,1,\dots,p^s-1$, \,$s\in J_{p;m}$; \,$\nu=1,2,\dots$. Thus we have
$$
D_{s}=\big(u_{s},A_{s}u_{s},\dots,A_{s}^{p^\nu-1}u_{s}\big)^T.
$$
Consequently, to describe all unitary matrixes $D_{s}$, we should find all vectors
$u_{s}=(\alpha_{s;0},\alpha_{s;1},\dots,\alpha_{s;p^\nu-1})^{T}$
such that the system of vectors $\{A_{s}^ru_{s}, r=0,\dots,p^\nu-1\}$ is orthonormal,
$s\in J_{p;m}$.
We already have such a vector $u_0=(1,0,\dots,0,0)^{T}$
because the matrix $D_0=\big(u_0,Au_0,\dots,A^{p^\nu-1}u_0\big)^T$
is the identity matrix.

Now we prove that the system
$\{A_{s}^ru_{s}, r=0,\dots,p^\nu-1\}$ is orthonormal if and only if
$u_{s}=B_{s}u_0$, where $B_{s}$ is a unitary matrix such that
$A_{s}B_{s}=B_{s}A_{s}$, \,$s\in J_{p;m}$. Indeed,
if $u_{s}=B_{s}u_0$, $B_{s}$ is a unitary matrix such that $A_{s}B_{s}=B_{s}A_{s}$,
then $A_{s}^ru_{s}=B_{s}A_{s}^ru_0$, \ $r=0,1,\dots,p^\nu-1$, \,$s\in J_{p;m}$.
Since the system $\{A_{s}^ru_0, r=0,1,\dots,p^\nu-1\}$ is orthonormal and the
matrix $B_{s}$ is unitary, the vectors $A_{s}^ru_{s}$, \ $r=0,1,\dots,p^\nu-1$
are also orthonormal, $s\in J_{p;m}$.
Conversely, if the system $A_{s}^ru_{s}$, \ $r=0,1,\dots,p^\nu-1$ is
orthonormal, taking into account that  $\{A_{s}^ru_0, r=0,1,\dots,p^\nu-1\}$
is also an orthonormal system, we conclude that
there exists a unitary matrix $B_{s}$ such that $A_{s}^ru_{s}=B_{s}(A_{s}^ru_0)$, \
$r=0,1,\dots,p^\nu-1$. Since $A_{s}^{p^\nu}u_{s}=\chi_p(-s)u_{s}$ and
$A_{s}^{p^\nu}u_0=\chi_p(-s)u_0$, we obtain additionally
$A_{s}^{p^\nu}u_{s}=B_{s}A_{s}^{p^\nu}u_0$. It follows
from the above relations that $(A_{s}B_{s}-B_{s}A_{s})(A_{s}^{r}u_0)=0$, \
$r=0,1,\dots,p^\nu-1$. Since the vectors $A_{s}^ru_0$, \
$r=0,1,\dots,p^\nu-1$ form a basis in the $p^\nu$-dimensional
space, we conclude that $A_{s}B_{s}=B_{s}A_{s}$, \,$s\in J_{p;m}$.

Thus all required unitary matrixes (\ref{105-11}) are given by
$$
D_{s}=\big(B_{s}u_0,B_{s}A_{\nu}u_0,\dots,B_{s}A_{s}^{p^\nu-1}u_0\big)^T,
$$
where $u_0=(1,0,\dots,0,0)^{T}$ and $B_{s}$ is a unitary matrix such that
$A_{s}B_{s}=B_{s}A_{s}$, \,$s\in J_{p;m}$.
It remains to describe all such matrixes $B_{s}$.

It is easy to see that the eigenvalues of $A_{s}$ and
the corresponding normalized eigenvectors are  respectively
\begin{equation}
\label{106-11}
\lambda_{s;r}=e^{2i\pi\frac{-s+r}{p^\nu}}
\end{equation}
and
$$
v_{s;r}=\big((v_{s;r})_0,\dots,(v_{s;r})_{p^\nu-1}\big)^T,
$$
where
\begin{equation}
\label{107-11}
(v_{s;r})_l=p^{-\nu/2}e^{-2i\pi\frac{-s+r}{p^\nu}l},
\quad l=0,1,2,\dots,p^\nu-1,
\end{equation}
$r=0,1,\dots,p^\nu-1$, \,$s\in J_{p;m}$.
Hence the matrix $A_{s}$ can be represented as
$A_{s}=C_{s}\widetilde{A}_{s}C_{s}^{-1}$, where
$$
\widetilde{A}_{s}=\left(
\begin{array}{ccccc}
\lambda_{0}&0&\ldots&0 \\
0&\lambda_{1}&\ldots&0 \\
\vdots&\vdots&\ddots&\vdots \\
0&0&\ldots&\lambda_{p^\nu-1} \\
\end{array}
\right)
$$
is a diagonal matrix, $C_{s}=\big(v_{s;0},\dots,v_{s;p^\nu-1}\big)$
is a unitary matrix. It follows that the matrix $B_{s}=C_{s}\widetilde{B}_{s}C_{s}^{-1}$
is unitary if and only if $\widetilde{B}_{s}$ is unitary. On the other hand,
$A_{s}B_{s}=B_{s}A_{s}$ if and only if
$\widetilde{A}_{s}\widetilde{B}_{s}=\widetilde{B}_{s}\widetilde{A}_{s}$.
Moreover, since according to (\ref{106-11}), $\lambda_{s;k}\ne\lambda_{s;l}$
whenever $k\ne l$, all unitary matrixes $\widetilde{B}_{s}$ such that
$\widetilde{A}_{s}\widetilde{B}_{s}=\widetilde{B}_{s}\widetilde{A}_{s}$, are given by
$$
\widetilde{B}_{s}=\left(
\begin{array}{ccccc}
\gamma_{s;0}&0&\ldots&0 \\
0&\gamma_{s;1}&\ldots&0 \\
\vdots&\vdots&\ddots&\vdots \\
0&0&\ldots&\gamma_{s;p^\nu-1} \\
\end{array}
\right),
$$
where $\gamma_{s;k}\in \bC$, \,$|\gamma_{s;k}|=1$, $k=0,1,\dots,p^\nu-1$, \,$s\in J_{p;m}$.
Hence all unitary matrixes $B_{s}$ such that $A_{s}B_{s}=B_{s}A_{s}$,
are given by $B_{s}=C\widetilde{B}_{s}C_{s}^{-1}$.
Using (\ref{107-11}), one can calculate
$$
\alpha_{s;k}=(B_{s}u_0)_k=(C_{s}\widetilde{B}_{s}C_{s}^{-1}u_0)_k
\qquad\qquad\qquad
$$
$$
\qquad\qquad
=\sum_{r=0}^{p^\nu-1}\gamma_{s;r}(v_{s;r})_k(\overline{v}_{s;r})_0
=p^{-\nu}\sum_{r=0}^{p^\nu-1}\gamma_{\nu;r}e^{-2i\pi\frac{-s+r}{p^\nu}k},
$$
where $\gamma_{s;k}\in \bC$, $|\gamma_{s;k}|=1$, $k=0,1,\dots,p^\nu-1,$, \,$s\in J_{p;m}$.

It remains to prove that
\begin{equation}
\label{109-11}
\{p^{-j/2}\psi_{s}^{(m)[\nu]}(p^{j}x-a), x\in\bQ_p: s\in J_{p;m}, \, j\in \bZ, \,a\in I_p\}
\end{equation}
is a basis for ${\cL}^2(\bQ_p)$ whenever $\psi_{s}^{(m)[\nu]}$ is defined by (\ref{101-11}),
(\ref{108-11}), \,$\nu=1,2,\dots$. Since according to Theorem~\ref{th2},
$$
\{p^{-j/2}\theta_{s}^{(m)}(p^{j}x-a), x\in\bQ_p: s\in J_{p;m},\, j\in \bZ, \,a\in I_p\}
$$
is a basis for ${\cL}^2(\bQ_p)$, it suffices to check that any function
$p^{-j/2}\theta_{s}^{(m)}(p^{j}x-c)$, $c\in I_p$, can be
decomposed with respect to the functions $p^{-j/2}\psi_{s}^{(m)[\nu]}(p^{j}x-a)$,
$a\in I_p$; where $s\in J_{p;m}$, $j\in \bZ$. Any $c\in I_p$, $c\ne 0$,
can be represented in the form $c=\frac{r}{p^\nu}+b$, where $r=0,1,\dots,p^\nu-1$,
$|b|_p\ge p^{\nu+1}$, $bp^\nu\in I_p$. Taking into account that
$\Xi_{s}^{[0]}=D_{s}^{-1}\Xi_{s}^{[\nu]}$, i.e.,
$$
\theta_{s}^{(m)}\Big(p^{j}x-\frac{r}{p^\nu}\Big)
=\sum_{k=0}^{p^\nu-1}\beta_{s;k}^{(r)}\psi_{s}^{(m)[\nu]}\Big(p^{j}x-\frac{k}{p^\nu}\Big),
\quad r=0,1,\dots,p^\nu-1,
$$
we have
$$
\theta_{s}^{(m)}\Big(p^{j}x-c\Big)=\theta_{s}^{(m)}\Big(p^{j}x-\frac{r}{p^\nu}-b\Big)=
\sum_{k=0}^{p^\nu-1}\beta_{s;k}^{(r)}\psi_{s}^{(m)[\nu]}\Big(p^{j}x-\frac{k}{p^\nu}-b\Big),
$$
and $\frac{k}{p^\nu}+b\in I_p$, $k=0,1,\dots,p^\nu-1$, \,$\nu=1,2,\dots$. Consequently, any function
$f\in {\cL}^2(\bQ_p)$ can be decomposed with respect to the system of functions (\ref{109-11}).
\end{proof}

Thus, we have constructed a countable family of non-Haar wavelet bases
given by formulas (\ref{109-11}), (\ref{101-11}), (\ref{108-11}).

\begin{Remark}
\label{rem2} \rm
Let $\nu=1$. According to (\ref{109-11}), (\ref{101-11}), (\ref{108-11});
(\ref{62.1-11}), (\ref{62.1}), in this case we have
$$
\psi_{s}^{(m)[1]}(x-a)
=\sum_{k=0}^{p-1}\alpha_{s;k}\theta_{s}^{(m)}\Big(x-a-\frac{k}{p}\Big),
\quad s\in J_{p;m}, \quad a\in I_p.
$$
Applying formulas (\ref{o-64.7}), (\ref{14}) to the last relation, we obtain
$$
F[\psi_{s}^{(m)[1]}(x-a)](\xi)
=\chi_p(a\xi)\Omega\big(|\xi+s|_p\big)\sum_{k=0}^{p-1}\alpha_{s;k}\chi_p\Big(\frac{k}{p}\xi\Big),
\quad s\in J_{p;m}, \quad a\in I_p.
$$
Thus the right-hand side of the last relation is not equal to zero only if $\xi=-s+\eta$,
$\eta=\eta_0+\eta_1p+\cdots \in \bZ_p$, i.e., $\xi \in B_0(-s)$.
Hence, for $\xi \in B_0(-s)$ we have
$$
F[\psi_{s}^{(m)[1]}(x-a)](\xi)
=\chi_p(-as)\chi_p(a\eta)\sum_{k=0}^{p-1}\alpha_{s;k}\chi_p\Big(\frac{k}{p}(-s+\eta_0)\Big),
\quad \eta\in \bZ_p,
$$
where $s\in J_{p;m}$, $a\in I_p,$.
It is easy to verify that for any $s\in J_{p;m}$ and $a\in I_p$ the right-hand
side of last relation {\em is not a characteristic function of any set}.
According to~\cite[Proposition~5.1.]{Ben-Ben}, {\em only} the functions whose Fourier
transforms are the characteristic functions of some sets may be wavelet
functions obtained by Benedettos' method. Consequently, the wavelet basis corresponding
to the generating wavelet functions $\psi_{s}^{(m)[1]}$, $s\in J_{p;m}$,
{\em cannot be constructed} by the method of~\cite{Ben-Ben}.
\end{Remark}

\subsection{Multidimensional non-Haar $p$-adic wavelets.}\label{s3.3}
Since the one- dimensional wavelets (\ref{62.1}) are non-Haar type,
we cannot construct the $n$-dimensional wavelet basis as the tensor products
of the one-dimensional MRAs (see~\cite{S-Sk-1}).
In this case we introduce $n$-dimensional non-Haar
wavelet functions as the $n$-direct product of the one-dimensional
non-Haar wavelets (\ref{62.1}). For $x=(x_1,\dots,x_n)\in \bQ_p^n$ and
$j=(j_1,\dots,j_n)\in \bZ^n$ we introduce a multi-dilatation
\begin{equation}
\label{62.8-md}
\widehat{p^{j}}x\stackrel{def}{=}(p^{j_1}x_1,\dots,p^{j_n}x_n)
\end{equation}
and define the $n$-direct products of the one-dimensional $p$-adic
wavelets (\ref{62.1}) as
\begin{equation}
\label{62.8}
\Theta_{s;\,j a}^{(m)\times}(x)=p^{-|j|/2}
\chi_p\big(s\cdot(\widehat{p^{j}}x-a)\big)
\Omega\big(|\widehat{p^{j}}x-a|_p\big), \quad x\in \bQ_p^n,
\end{equation}
where $j=(j_1,\dots,j_n)\in \bZ^n$; $|j|=j_1+\cdots+j_n$;
$a=(a_1,\dots,a_n)\in I_p^n$;
$s=(s_1,\dots,s_n)\in J_{p;m}^n$; $m=(m_1,\dots,m_n)$, $m_l\ge 1$
is a {\it fixed} positive integer, $l=1,2,\dots,n$. Here $I_p^n$, $J_{p;m}^n$
are the $n$- direct products of the corresponding sets (\ref{62.0**}) and
(\ref{62.0*}).

In view of (\ref{9}), Theorem~\ref{th2} implies the following
statement.

\begin{Theorem}
\label{th2.1}
The non-Haar wavelet functions {\rm(\ref{62.8})} form an orthonormal
basis in ${\cL}^2(\bQ_p^n)$.
\end{Theorem}

Using (\ref{9}), it is easy to verify that
$\int_{\bQ_p^n}\Theta_{s;\, j a}^{(m)\times}(x)\,d^nx=0$,
i.e., in view of Lemma~\ref{lem1}, the wavelet function
$\Theta_{s;\, j a}^{(m)\times}(x)$, $j\in \bZ^n$, $a\in I_p^n$, $s\in J_{p;m}^n$,
belongs to the Lizorkin space $\Phi(\bQ_p^n)$.

\begin{Corollary}
\label{cor3-1}
The $n$-direct product of one-dimensional Kozyrev's wavelets
{\rm (\ref{62.0-1})}
\begin{equation}
\label{64**}
\Theta_{k;j a}^{\times}(x)
=p^{-|j|/2}\chi_p\big(p^{-1}k\cdot(\widehat{p^{j}}x-a)\big)
\Omega\big(|\widehat{p^{j}}x-a|_p\big), \quad x\in \bQ_p^n,
\end{equation}
form an orthonormal complete basis in ${\cL}^2(\bQ_p^n)$, \ $k \in J_p^n$, \
$j\in \bZ^n$, $a\in I_p^n$.
\end{Corollary}

The proof follows from Theorem~\ref{th2.1} if we set $m=1$.

\begin{Corollary}
\label{cor2}
The family of functions
\begin{equation}
\label{o-64.8*}
{\widetilde\Theta}_{s;\, j a}^{(m)\times}(\xi)
=F[\Theta_{s;\, j a}^{(m)\times}](\xi)
=p^{|j|/2}\chi_p\big(\widehat{p^{-j}}a\cdot\xi\big)
\Omega\big(|s+\widehat{p^{-j}}\xi|_p\big),\quad \xi\in \bQ_p^n,
\end{equation}
form an orthonormal complete basis in ${\cL}^2(\bQ_p^n)$, $j\in \bZ^n$;
$a\in I_p^n$; $s\in J_{p;m}^n$; $m=(m_1,\dots,m_n)$, $m_l\ge 1$ is
a {\it fixed} positive integer, $l=1,2,\dots,n$.
\end{Corollary}

\begin{proof}
Consider the function $\Theta_{s}^{(m)\times}(x)=\chi_p(s\cdot x)\Omega(|x|_p)$
generated by the direct product of functions (\ref{62.1-11}),
$x\in \bQ_p^n$, $s=(s_1,\dots,s_n)\in J_{p;m}^n$, $s_k\in J_{p;m_k}$,
$k=1,2,\dots,n$. Using (\ref{9}), (\ref{14.1}), (\ref{14}), we have
$$
F[\Theta_{s}^{(m)\times}(x)](\xi)
=F\Big[\prod_{k=1}^n\chi_p(x_ks_k)\Omega(|x_k|_p)\Big](\xi)
=\prod_{k=1}^nF\Big[\Omega(|x_k|_p)\Big](\xi_k+s_k|_p)
$$
\begin{equation}
\label{o-64.7}
=\prod_{k=1}^n\Omega\big(|\xi_k+s_{k}|_p\big)
=\Omega\big(|\xi+s|_p\big), \quad \xi\in \bQ_p^n.
\end{equation}
Here, in view of (\ref{9}), $\Omega\big(|\xi+s|_p\big)
=\Omega\big(|\xi_1+s_{1}|_p\big)\times\cdots\times
\Omega\big(|\xi_n+s_{n}|_p\big)$.
According to (\ref{62.0*}), $|s_{k}|_p=p^{m_k}$ and
$\Omega\big(|\xi_k+s_{k}|_p\big)\ne 0$ only if
$\xi_k=-s_{k}+\eta_k$, where $\eta_k\in \bZ_p$,
$s_k \in J_{p;m_k}$, \ $k=1,2,\dots,n$. This yields $\xi=-s+\eta$,
where $\eta \in \bZ_p^n$, $s\in J_{p;m}^n$, and in view
of (\ref{8}), $|\xi|_p=p^{\max\{m_1,\dots,m_n\}}$.

In view of formulas (\ref{62.8}), (\ref{o-64.7}), (\ref{14}), we have
$$
F[\Theta_{s;\, j a}^{(m)\times}(x)](\xi)
=p^{-|j|/2}F[\Theta_{s}^{(m)\times}(\widehat{p^{j}}x-a)](\xi)
\qquad\qquad\qquad\qquad\quad
$$
$$
=p^{|j|/2}\chi_p\big(\widehat{p^{-j}}a\cdot\xi\big)
\Omega\big(|s+\widehat{p^{-j}}\xi|_p\big),
$$
i.e., (\ref{o-64.8*}).

The formula (\ref{o-64.8*}), the Parseval formula~\cite[VII,(4.1)]{Vl-V-Z}, and
Theorem~\ref{th2.1} imply the statement.
\end{proof}

Similarly, one can construct $n$-dimensional non-Haar wavelet bases generated
by the one-dimensional non-Haar wavelets (\ref{109-11}) (as the $n$-direct product):
\begin{equation}
\label{62.8-1}
\Psi_{s;\,j a}^{(m)[\nu]\times}(x)=p^{-|j|/2}\psi_{s}^{(m_1)[\nu]}(p^{j_1}x_1-a_1)\cdots
\psi_{s}^{(m_n)[\nu]}(p^{j_n}x_n-a_n),
\end{equation}
where $x\in \bQ_p^n$ and $\psi_{s_{j_k}}^{(m_{j_k})[\nu]}$ is defined by (\ref{101-11}),
(\ref{108-11}), $j=(j_1,\dots,j_n)\in \bZ^n$; $|j|=j_1+\cdots+j_n$;
$a=(a_1,\dots,a_n)\in I_p^n$;
$s=(s_1,\dots,s_n)\in J_{p;m}^n$; $m=(m_1,\dots,m_n)$, $m_k\ge 1$ is
a {\it fixed} positive integer, $k=1,2,\dots,n$; \,$\nu=1,2,\dots$.

\subsection{$p$-Adic Lizorkin spaces and wavelets.}\label{s3.4}
Now we prove an analog of Lemma~\ref{lem-four-1} for the Lizorkin
test functions from $\Phi(\bQ_p^n)$.

\begin{Lemma}
\label{lem-w-1**}
Any function $\phi \in \Phi(\bQ_p^n)$ can be represented in the form of a {\em finite} sum
\begin{equation}
\label{wav-9.4=1}
\phi(x)=\sum_{s\in J_{p;m}^n,j\in \bZ^n,a\in I_p^n}c_{s;\,j, a}
\Theta_{s;\,j a}^{(m)\times}(x), \quad x \in \bQ_p^n,
\end{equation}
where $c_{s;\,j, a}$ are constants; $\Theta_{s;\,j a}^{(m)\times}(x)$ are elements
of the non-Haar wavelet basis {\rm(\ref{62.8})}; $s=(s_1,\dots,s_n)\in J_{p;m}^n$;
$j=(j_1,\dots,j_n)\in \bZ^n$, $|j|=j_1+\cdots+j_n$; $a=(a_1,\dots,a_n)\in I_p^n$;
$m=(m_1,\dots,m_n)$, $m_l\ge 1$ is a {\it fixed} positive integer, $l=1,2,\dots,n$.
\end{Lemma}

\begin{proof}
Let us calculate $\cL^2(\bQ_p^n)$-scalar product $(\phi(x),\Theta_{s;\,j a}^{(m)\times}(x))$.
Taking into account formula (\ref{o-64.8*}) and using the Parseval-Steklov theorem,
we obtain
$$
c_{s;\,j, a}=\big(\phi(x),\Theta_{s;\,j a}^{(m)\times}(x)\big)
=\big(F[\phi](\xi),F[\Theta_{s;\,j a}^{(m)\times}](\xi)\big)
$$
\begin{equation}
\label{wav-9.4=2}
=\big(\psi(\xi), p^{|j|/2}\chi_p\big(\widehat{p^{-j}}a\cdot\xi\big)
\Omega\big(|s+\widehat{p^{-j}}\xi|_p\big)\big),
\end{equation}
where $j\in \bZ^n$, $a\in I_p^n$, $s\in J_{p;m}^n$. Here,
according to Lemma~\ref{lem1}, any function
$\phi \in \Phi(\bQ_p^n)$ belongs to one of the spaces ${\cD}^l_N(\bQ_p^n)$,
$\psi=F^{-1}[\phi]\in \Psi(\bQ_p^n)\cap \cD^{-N}_{-l}(\bQ_p^n)$, and
${\rm supp}\,\psi\subset B^n_{-l}\setminus B^n_{-N}$.

Let $|s|_p\ne |\widehat{p^{-j}}\xi|_p$. Since $p^{-N}\le|\xi|_p\le p^{-l}$
and
$$
|s+\widehat{p^{-j}}\xi|_p
=\max(|s|_p,|\widehat{p^{-j}}\xi|_p)=\max\big(p^{\max(m_1,\dots,m_n)},\max(p^{j_k}|\xi_k|_p\big),
$$
in view of (\ref{wav-9.4=2}), it is clear that there are finite quantity of indexes
$s=(s_1,\dots,s_n)\in J_{p;m}^n$, $j=(j_1,\dots,j_n)\in \bZ^n$ such that $c_{s;\,j, a}\ne 0$.
The case $|s|_p=|\widehat{p^{-j}}\xi|_p=p^{\max(m_1,\dots,m_n)}$ can be considered
in the same way.

Thus equality (\ref{wav-9.4=1}) holds in the sense of ${\cL}^2(\bQ_p^n)$. Consequently,
it holds in the usual sense.
\end{proof}

Using standard results from the book~\cite{Schaefer} or repeating the
reasoning~\cite{Kh-Sh-Sm1}, \cite{Kh-Sh-Sm2} almost word for word,
we obtain the following assertion.

\begin{Proposition}
\label{pr-w-2**}
Any distribution $f \in \Phi'(\bQ_p^n)$ can be realized in the form of an
{\em infinite} sum of the form
\begin{equation}
\label{wav-9.4=3}
f(x)=\sum_{s\in J_{p;m}^n,j\in \bZ^n,a\in I_p^n}d_{s;\,j, a}
\Theta_{s;\,j a}^{(m)\times}(x), \quad x \in \bQ_p^n,
\end{equation}
where $d_{s;\,j, a}$ are constants; $\Theta_{s;\,j a}^{(m)\times}(x)$ are elements
of the non-Haar wavelet basis {\rm(\ref{62.8})}; $s=(s_1,\dots,s_n)\in J_{p;m}^n$;
$j=(j_1,\dots,j_n)\in \bZ^n$, $|j|=j_1+\cdots+j_n$; $a=(a_1,\dots,a_n)\in I_p^n$;
$m=(m_1,\dots,m_n)$, $m_l\ge 1$ is a {\it fixed} positive integer, $l=1,2,\dots,n$.
\end{Proposition}

Here any distribution $f \in \Phi'(\bQ_p^n)$ is associated with the representation
(\ref{wav-9.4=3}), where the coefficients
\begin{equation}
\label{wav-9.4=4}
d_{s;\,j, a}\stackrel{def}{=}\bigl\langle f,\Theta_{s;\,j a}^{(m)\times}\bigr\rangle,
\quad s\in J_{p;m}^n, \quad j\in \bZ^n, \quad a\in I_p^n.
\end{equation}
And vice versa, taking into account
Lemma~\ref{lem-w-1**} and orthonormality of the wavelet basis (\ref{62.8}), any
infinite sum is associated with the distribution $f \in \Phi'(\bQ_p^n)$ whose
action on a test function $\phi \in \Phi(\bQ_p^n)$ is defined as
\begin{equation}
\label{wav-9.4=5}
\langle f,\phi\rangle=\sum_{s\in J_{p;m}^n,j\in \bZ^n,a\in I_p^n}d_{s;\,j, a}c_{s;\,j, a},
\end{equation}
where the sum is finite.

It is clear that in Lemma~\ref{lem-w-1**} and Proposition~\ref{pr-w-2**}
instead of the basis (\ref{62.1}) or its multidimensional generalization
(\ref{62.8}), one can use the bases (\ref{109-11}), (\ref{101-11}), (\ref{108-11})
and their multidimensional generalizations.

In~\cite{Al-Koz}, the assertions of the type of
Lemma~\ref{lem-w-1**} and Proposition~\ref{pr-w-2**} were stated for ultrametric
Lizorkin spaces.

\section{Spectral theory of $p$-adic pseudo-differential operators}
\label{s4}

\subsection{Pseudo-differential operators in the Lizorkin spaces.}\label{s4.1}
In this subsection we present some facts on pseudo-differential operators which
were introduced in~\cite{Al-Kh-Sh3},~\cite{Al-Kh-Sh4}. Consider a class of
pseudo-differential operators in the Lizorkin space $\Phi(\bQ_p^n)$
$$
(A\phi)(x)=F^{-1}\big[\cA(\cdot)\,F[\phi](\cdot)\big](x)
\qquad\qquad\qquad\qquad\qquad\qquad\qquad\qquad\qquad
$$
\begin{equation}
\label{64.3}
=\int_{\bQ_p^n}\int_{\bQ_p^n}\chi_p\big((y-x)\cdot \xi\big)
\cA(\xi)\phi(y)\,d^n\xi\,d^ny,
\quad \phi \in \Phi(\bQ_p^n)
\end{equation}
with symbols $\cA\in \cE(\bQ_p^n\setminus \{0\})$.

\begin{Remark}
\label{rem3} \rm
The class of operators (\ref{64.3}) includes the Taibleson fractional
operator with the symbol of the form $|\xi|^{\alpha}_p$ (see (\ref{59**}));
the Kochubei operator with the symbol of the form
$\cA(\xi)=|f(\xi_1,\dots,\xi_n)|_p^{\alpha}$, $\alpha> 0$, where
$f(\xi_1,\dots,\xi_n)$ is a quadratic form such that
$f(\xi_1,\dots,\xi_n)\ne 0$ when $|\xi_1|_p+\cdots |\xi_n|_p\ne 0$
(see~\cite{Koch3},~\cite{Koch4}); the Zuniga-Galindo operator
with the symbol of the form $\cA(\xi)=|f(\xi_1,\dots,\xi_n)|_p^{\alpha}$,
$\alpha> 0$, where $f(\xi_1,\dots,\xi_n)$ is a non-constant polynomial
(see~\cite{Z1},~\cite{Z2}).
\end{Remark}

If we define a conjugate pseudo-differential operator $A^{T}$ as
$$
(A^{T}\phi)(x)=F^{-1}[\cA(-\xi)F[\phi](\xi)](x)
=\int_{\bQ_p^n}\chi_p(-x\cdot \xi)\cA(-\xi)F[\phi](\xi)\,d^n\xi
$$
then one can define the operator $A$ in the Lizorkin space of
distributions: for $f \in \Phi'(\bQ_p^n)$ we have
\begin{equation}
\label{64.4}
\langle Af,\phi\rangle=\langle f,A^{T}\phi\rangle,
\qquad \forall \, \phi \in \Phi(\bQ_p^n).
\end{equation}

\begin{Lemma}
\label{lem4}
{\rm(~\cite{Al-Kh-Sh3})}
The Lizorkin spaces $\Phi(\bQ_p^n)$ and $\Phi'(\bQ_p^n)$ are invariant under the
pseudo-differential operators {\rm(\ref{64.3})}.
\end{Lemma}

\begin{proof}
In view of (\ref{12}) and results of Subsec.~\ref{s2.2}, both functions
$F[\phi](\xi)$ and $\cA(\xi)F[\phi](\xi)$ belong to $\Psi(\bQ_p^n)$,
and, consequently, $(A\phi)(x)\in \Phi(\bQ_p^n)$, i.e.,
$A(\Phi(\bQ_p^n))\subset \Phi(\bQ_p^n)$. Thus the
pseudo-differential operators (\ref{64.3}) are well defined, and the
Lizorkin space $\Phi(\bQ_p^n)$ is invariant under them. In view of
(\ref{64.4}), the Lizorkin space of distributions $\Phi'(\bQ_p^n)$
is also invariant under pseudo-differential operator $A$.
\end{proof}

\subsection{The Taibleson fractional operator in the Lizorkin spaces.}\label{s4.2}
In particular, setting $\cA(\xi)=|\xi|_p^{\alpha}$, $\xi\in \bQ_p^n$, we obtain
the multi-dimensional Taibleson fractional operator. This operator was introduced
in~\cite[\S2]{Taib1},~\cite[III.4.]{Taib3} on the space of distributions
${\cD}'(\bQ_p^n)$ for $\alpha\in \bC$, $\alpha\ne -n$. Next, in~\cite{Al-Kh-Sh3},
the Taibleson fractional operator was defined and studied in the Lizorkin space of
distributions $\Phi'(\bQ_p^n)$ for all $\alpha\in \bC$.
According to (\ref{64.3}), (\ref{64.4}),
\begin{equation}
\label{61**}
\big(D^{\alpha}f\big)(x)
=F^{-1}\big[|\cdot|^{\alpha}_pF[f](\cdot)\big](x),
\quad f \in \Phi'(\bQ_p^n).
\end{equation}
Representation (\ref{61**}) can be rewritten as a convolution
\begin{equation}
\label{59**}
\big(D^{\alpha}f\big)(x)=\kappa_{-\alpha}(x)*f(x)
=\langle \kappa_{-\alpha}(x),f(x-\xi)\rangle,
\quad f\in \Phi'(\bQ_p^n), \quad \alpha \in \bC,
\end{equation}
where according to~\cite{Al-Kh-Sh3}, the multidimensional {\it Riesz kernel\/}
is given by the formula
$$
\kappa_{\alpha}(x)=\left\{
\begin{array}{lll}
\frac{|x|_p^{\alpha-n}}{\Gamma_p^{(n)}(\alpha)}, &&
\alpha \ne 0, \, \,  n, \\
\delta(x) && \alpha=0, \\
-\frac{1-p^{-n}}{\log p}\log|x|_p && \alpha=n \\
\end{array}
\right.
$$
the function $|x|_p$, \ $x\in \bQ_p^n$ is defined by (\ref{8}). Here the
multidimensional homogeneous distribution
$|x|_p^{\alpha-n}\in {\cD}'(\bQ_p^n)$ of degree $\alpha-n$ was defined
in~\cite[(*)]{Taib1},~\cite[III,(4.3)]{Taib3},~\cite[VIII,(4.2)]{Vl-V-Z},
$\Gamma^{(n)}_p(\alpha)$ is the $n$-dimensional $\Gamma$-{\it function\/}
defined in~\cite[Theorem~1.]{Taib1},~\cite[III,Theorem~(4.2)]{Taib3},
~\cite[VIII,(4.4)]{Vl-V-Z}.

According to Lemma~\ref{lem4} and (\ref{64.3}), (\ref{64.4}), the Lizorkin
space $\Phi'(\bQ_p^n)$ is invariant under the Taibleson fractional operator
$D^{\alpha}$  for all $\alpha \in \bC$~\cite{Al-Kh-Sh3}.

\subsection{$p$-Adic wavelets as eigenfunctions of pseudo-differential operators.}\label{s4.3}
As mentioned above in Sec.~\ref{s1}, it is typical that $p$-adic
compactly supported wavelets are eigenfunctions of $p$-adic pseudo-differential
operators. For example, in~\cite{Koz0} S.~V.~Kozyrev proved that wavelets
(\ref{62.0-1}) are eigenfunctions of the one-dimensional fractional operator
(\ref{61**}), (\ref{59**}) for $\alpha>0$:
\begin{equation}
\label{62.2-v}
D^{\alpha}\theta_{k;\,j a}(x)=p^{\alpha(1-j)}\theta_{k;\,j a}(x),
\quad x\in \bQ_p,
\end{equation}
where $k=1,2,\dots p-1$, \ $j\in \bZ$, $a\in I_p$.
Since wavelet functions (\ref{62.0-1}) belong to the Lizorkin space,
the relation (\ref{62.2-v}) holds for all $\alpha\in \bC$.

Now we study the spectral problem for pseudo-differential operators (\ref{64.3})
in connection with the wavelet functions (\ref{62.8}) and (\ref{62.8-1}).

\begin{Theorem}
\label{th4.1}
Let $A$ be a pseudo-differential operator {\rm(\ref{64.3})}
with a symbol $\cA(\xi)\in \cE(\bQ_p^n\setminus \{0\})$.
Then the $n$-dimensional non-Haar wavelet function {\rm (\ref{62.8})}
$$
\Theta_{s;\,j a}^{(m)\times}(x)=p^{-|j|/2}
\chi_p\big(s\cdot(\widehat{p^{j}}x-a)\big)
\Omega\big(|\widehat{p^{j}}x-a|_p\big), \quad x\in \bQ_p^n,
$$
is an eigenfunction of $A$ if and only if
\begin{equation}
\label{64.1***}
\cA\big(\widehat{p^{j}}(-s+\eta)\big)=\cA\big(-\widehat{p^{j}}s\big),
\qquad \forall \, \eta \in \bZ_p^n,
\end{equation}
where $j=(j_1,\dots,j_n)\in \bZ^n$; $a\in I_p^n$;
$s\in J_{p;m}^n$; and $m=(m_1,\dots,m_n)$, $m_j\ge 1$ is
a {\it fixed} positive integer, $j=1,2,\dots,n$.
The corresponding eigenvalue is $\lambda=\cA\big(-\widehat{p^{j}}s\big)$, i.e.,
$$
A\Theta_{s;\,j a}^{(m)\times}(x)
=\cA(-\widehat{p^{j}}s)\Theta_{s;\,j a}^{(m)\times}(x).
$$
Here the multi-dilatation is defined by {\rm(\ref{62.8-md})},
$I_p^n$ and $J_{p;m}^n$ are the $n$-direct products of the
corresponding sets {\rm(\ref{62.0**})} and {\rm(\ref{62.0*})}.
\end{Theorem}

\begin{proof}
Let condition (\ref{64.1***}) be satisfied. Then (\ref{64.3})
and the above formula (\ref{o-64.8*}) imply that
$$
A\Theta_{s;\,j a}^{(m)\times}(x)
=F^{-1}\big[\cA(\xi)F[\Theta_{s;\,j a}^{(m)\times}](\xi)\big](x)
\qquad\qquad\qquad\qquad\qquad\qquad\quad
$$
\begin{equation}
\label{o-64.9}
=p^{|j|/2}F^{-1}\big[\cA(\xi)\chi_p\big(\widehat{p^{-j}}a\cdot\xi\big)
\Omega\big(|s+\widehat{p^{-j}}\xi|_p\big)\big](x).
\end{equation}
Making the change of variables $\xi=\widehat{p^{j}}(\eta-s)$
and using (\ref{14.1}), we obtain
$$
A\Theta_{s;\,j a}^{(m)\times}(x)
=p^{-|j|/2}
\int\limits_{\bQ_p^n}\chi_p\big(-(\widehat{p^{j}}x-a)\cdot (\eta-s)\big)
\cA(\widehat{p^{j}}(\eta-s))\,\Omega(|\eta|_p)\,d^n\eta
$$
$$
\qquad\qquad\quad
=p^{-|j|/2}\cA(-\widehat{p^{j}}s)
\chi_p\big(s\cdot(\widehat{p^{j}}x-a)\big)
\int_{B_{0}^n}\chi_p(-(\widehat{p^{j}}x-a)\cdot\eta)\,d^n\eta
$$
$$
=\cA(-\widehat{p^{j}}s)\Theta_{s;\,j a}^{(m)}(x).
\qquad\qquad\qquad\qquad\qquad\qquad
$$
Consequently,
$A\Theta_{s;\,j a}^{(m)\times}(x)=\lambda\Theta_{s;\,j a}^{(m)\times}(x)$,
where $\lambda=\cA(-\widehat{p^{j}}s)$.

Conversely, if
$A\Theta_{s;\,j a}^{(m)\times}(x)=\lambda\Theta_{s;\,j a}^{(m)\times}(x)$,
$\lambda\in \bC$, taking the Fourier transform of
both left- and right-hand sides of this identity and
using (\ref{64.3}), (\ref{o-64.8*}), (\ref{o-64.9}), we have
\begin{equation}
\label{o-64.9==}
\big(\cA(\xi)-\lambda\big)\chi_p\big(\widehat{p^{-j}}a\cdot\xi\big)
\Omega\big(|s+\widehat{p^{-j}}\xi|_p\big)=0, \quad \xi\in \bQ_p^n.
\end{equation}
Now, if $s+\widehat{p^{-j}}\xi=\eta$, $\eta \in \bZ_p^n$, then
$\xi=\widehat{p^{j}}(-s+\eta)$. Since
$\chi_p\big(\widehat{p^{-j}}a\cdot\xi\big)\ne0$,
$\Omega\big(|s+\widehat{p^{-j}}\xi|_p\big)\ne0$,
it follows from (\ref{o-64.9==}) that
$\lambda=\cA\big(\widehat{p^{j}}(-s+\eta)\big)$ for any
$\eta \in\bZ_p^n$. Thus $\lambda=\cA(-\widehat{p^{j}}s)$, and,
consequently, (\ref{64.1***}) holds.

The proof of the theorem is complete.
\end{proof}

\begin{Corollary}
\label{cor3}
Let $A$ be a pseudo-differential operator {\rm(\ref{64.3})}
with the symbol $\cA(\xi)\in \cE(\bQ_p^n\setminus \{0\})$.
Then the $n$-dimensional wavelet function {\rm (\ref{64**})}
$$
\Theta_{k;j a}^{\times}(x)
=p^{-|j|/2}\chi_p\big(p^{-1}k\cdot(\widehat{p^{j}}x-a)\big)
\Omega\big(|\widehat{p^{j}}x-a|_p\big), \quad x\in \bQ_p^n,
$$
is an eigenfunction of $A$ if and only if
$$
\cA\big(\widehat{p^{j}}(-p^{-1}k+\eta)\big)
=\cA\big(-\widehat{p^{j-I}}k\big),
\qquad \forall \, \eta \in \bZ_p^n,
$$
where $k \in J_p^n$, $j\in \bZ^n$, $a\in I_p^n$, $I=(1,\dots,1)$.
The corresponding eigenvalue is $\lambda=\cA\big(-\widehat{p^{j-I}}j\big)$,
i.e.,
$$
A\Theta_{k;j a}^{\times}(x)=\cA(-\widehat{p^{j-I}}k)
\Theta_{k;j a}^{\times}(x).
$$
\end{Corollary}

Representation (\ref{101-11}) and Theorem~\ref{th4.1} imply the following statement.

\begin{Theorem}
\label{th4.1-1}
Let $A$ be a pseudo-differential operator {\rm(\ref{64.3})}
with the symbol $\cA(\xi)\in \cE(\bQ_p^n\setminus \{0\})$.
Then the $n$-dimensional wavelet function {\rm (\ref{62.8-1})}
$$
\Psi_{s;\,j a}^{(m)[\nu]\times}(x)=p^{-|j|/2}\psi_{s}^{(m_1)[\nu]}(p^{j_1}x_1-a_1)\cdots
\psi_{s}^{(m_n)[\nu]}(p^{j_n}x_n-a_n), \quad x\in \bQ_p^n,
$$
is an eigenfunction of $A$ if and only if condition {\rm(\ref{64.1***})} holds,
where $\psi_{s_{j_k}}^{(m_{j_k})[\nu]}$ is defined by {\rm(\ref{101-11})},
{\rm(\ref{108-11})}, $j=(j_1,\dots,j_n)\in \bZ^n$; $|j|=j_1+\cdots+j_n$;
$a=(a_1,\dots,a_n)\in I_p^n$;
$s=(s_1,\dots,s_n)\in J_{p;m}^n$; $m=(m_1,\dots,m_n)$, $m_k\ge 1$ is
a {\it fixed} positive integer, $k=1,2,\dots,n$; \,$\nu=1,2,\dots$.
The corresponding eigenvalue is $\lambda=\cA(-\widehat{p^{j}}s)$,
i.e.,
$$
A\Psi_{s;\,j a}^{(m)[\nu]\times}(x)=\cA(-\widehat{p^{j}}s)
\Psi_{s;\,j a}^{(m)[\nu]\times}(x).
$$
\end{Theorem}

\subsection{$p$-Adic wavelets as eigenfunctions of the Taibleson fractional operator.}\label{s4.4}
As mentioned above, the Taibleson fractional operator $D^{\alpha}$
has the symbol $\cA(\xi)=|\xi|_p^{\alpha}$.
The symbol $\cA(\xi)=|\xi|_p^{\alpha}$ satisfies the condition (\ref{64.1***}):
$$
\cA\big(\widehat{p^{j}}(-s+\eta)\big)
=|\widehat{p^{j}}(-s+\eta)|_p^{\alpha}
=\Big(\max_{1\le r\le n}\big(p^{-j_r}|-s_r|_p\big)\Big)^{\alpha}
\qquad
$$
$$
\quad
=\cA\big(-\widehat{p^{j}}s\big)
=p^{\alpha\max_{1\le r\le n}\{m_r-j_r\}}
$$
for all $\eta \in \bZ_p^n$. Consequently, according to Theorem~\ref{th4.1},
we have

\begin{Corollary}
\label{cor5}
The $n$-dimensional non-Haar $p$-adic wavelet {\rm(\ref{62.8})} is an
eigenfunction of the Taibleson fractional operator {\rm (\ref{59**})}:
$$
D^{\alpha}\Theta_{s;\,j a}^{(m)\times}(x)
=p^{\alpha\max_{1\le r\le n}\{m_r-j_r\}}
\Theta_{s;\,j a}^{(m)\times}(x), \quad \alpha \in\bC, \qquad x\in \bQ_p^n,
$$
$s \in J_{p;m}^n$, $j\in \bZ^n$, $a\in I_p^n$.
\end{Corollary}

In particular, in view of Corollary~\ref{cor3}, we have

\begin{Corollary}
\label{cor6}
The $n$-dimensional $p$-adic wavelet {\rm(\ref{64**})} is an
eigenfunction of the Taibleson fractional operator {\rm (\ref{59**})}:
$$
D^{\alpha}\Theta_{k;\, j a}^{\times}(x)
=p^{\alpha(1-\min_{1\le r\le n}j_r)}\Theta_{k;\, j a}^{\times}(x),
\quad \alpha \in\bC, \quad x\in \bQ_p^n,
$$
$k \in J_p^n$, $j\in \bZ^n$, $a\in I_p^n$.
\end{Corollary}

\begin{Corollary}
\label{cor7}
The $n$-dimensional $p$-adic wavelet {\rm(\ref{62.8-1})} is an
eigenfunction of the Taibleson fractional operator {\rm (\ref{59**})}:
$$
D^{\alpha}\Psi_{s;\,j a}^{(m)[\nu]\times}(x)
=p^{\alpha\max_{1\le r\le n}\{m_r-j_r\}}\Psi_{s;\,j a}^{(m)[\nu]\times}(x),
\quad \alpha \in\bC, \quad x\in \bQ_p^n,
$$
$s \in J_{p;m}^n$, $j\in \bZ^n$, $a\in I_p^n$,
\,$\nu=1,2,\dots$.
\end{Corollary}

\section{Application of $p$-adic wavelets to evolutionary pseudo-differential equations}
\label{s5}

\subsection{Linear equations.}\label{s5.1}
{\bf (a)} Let us consider the Cauchy problem for the {\em linear evolutionary
pseudo-differential equation}
\begin{equation}
\label{eq-70}
\left\{
\begin{array}{rclrcl}
\frac{\partial u(x,t)}{\partial t}+A_xu(x,t)&=&0,
&&\text{in} \quad \bQ_p^n\times (0, \ \infty), \medskip \\
u(x,t)&=&u^0(x),
&&\text{in} \quad \bQ_p^n\times \{t=0\}, \\
\end{array}
\right.
\end{equation}
where $t\in \bR$, $u^0\in \Phi'(\bQ_p^n)$ and
\begin{equation}
\label{eq-70-op}
A_xu(x,t)=F^{-1}\big[\cA(\xi)\,F[u(\cdot,t)](\xi)\big](x)
\end{equation}
is a pseudo-differential operator (\ref{64.3}) (with
respect to $x$) with symbols $\cA(\xi)\in \cE(\bQ_p^n\setminus \{0\})$,
$u(x,t)$ is the desired distribution such that $u(x,t)\in \Phi'(\bQ_p^n)$
for any $t\ge 0$.

In particular, we will consider the Cauchy problem
\begin{equation}
\label{eq-70-c}
\left\{
\begin{array}{rclrcl}
\frac{\partial u(x,t)}{\partial t}+D^{\alpha}_xu(x,t)&=&0,
&&\text{in} \quad \bQ_p^n\times (0, \ \infty), \medskip \\
u(x,t)&=&u^0(x),
&&\text{in} \quad \bQ_p^n\times \{t=0\}, \\
\end{array}
\right.
\end{equation}
where $D^{\alpha}_xu(x,t)=F^{-1}\big[|\xi|_p^{\alpha}\,F[u(\cdot,t)](\xi)\big](x)$
is the Taibleson fractional operator (\ref{61**})
with respect to $x$, \ $\alpha\in\bC$.

\begin{Theorem}
\label{eq-th5.1}
The Cauchy problem {\rm (\ref{eq-70})} has a unique solution
\begin{equation}
\label{eq-71}
u(x,t)=F^{-1}\big[F[u^{0}(\cdot)](\xi)e^{-\cA(\xi)\,t}\big](x).
\end{equation}
\end{Theorem}

\begin{proof}
Since $u(x,t)$ is a distribution such that $u(x,t)\in \Phi'(\bQ_p^n)$
for any $t\ge 0$, the relation (\ref{eq-70}) is well-defined. Applying
the Fourier transform to (\ref{eq-70}), we obtain the following equation
$$
\frac{\partial F[u(\cdot,t)](\xi)}{\partial t}+\cA(\xi)\,F[u(\cdot,t)](\xi)=0.
$$
Solving this equation, we obtain
$$
F[u(x,t)](\xi)=F[u(x,0)](\xi)e^{-\cA(\xi)\,t}.
$$
This implies (\ref{eq-71}).
\end{proof}

\begin{Theorem}
\label{eq-th5}
Let a pseudo-differential operator $A_{x}$ in {\rm (\ref{eq-70})} be
such that its symbol $\cA(\xi)$ satisfies the condition {\rm (\ref{64.1***})}:
$$
\cA\big(\widehat{p^{j}}(-s+\eta)\big)=\cA\big(-\widehat{p^{j}}s\big),
\qquad \forall \, \eta \in \bZ_p^n,
$$
for any $j\in \bZ^n$, $s\in J_{p;m}^n$.
Then the Cauchy problem {\rm (\ref{eq-70})} has a unique solution
\begin{equation}
\label{eq-75}
u(x,t)=\sum_{s\in J_{p;m}^n,j\in \bZ^n,a\in I_p^n}
\bigl\langle u^{0}(x),\Theta_{s;\,j a}^{(m)\times}\bigr\rangle
e^{-\cA(-\widehat{p^{j}}s)t}\Theta_{s;\,j a}^{(m)\times}(x),
\end{equation}
for $t\ge 0$, where $\Theta_{s;\,j a}^{(m)\times}(x)$ are $n$-dimensional
$p$-adic wavelets {\rm (\ref{62.8})}.
\end{Theorem}

\begin{proof}
According to the formula (\ref{wav-9.4=3}) from Proposition~\ref{pr-w-2**},
we will seek a solution of the Cauchy problem (\ref{eq-70}) in the form
of an infinite sum
\begin{equation}
\label{eq-73}
u(x,t)=\sum_{s\in J_{p;m}^n,j\in \bZ^n,a\in I_p^n}\Lambda_{s;\,j, a}(t)
\Theta_{s;\,j a}^{(m)\times}(x),
\end{equation}
where $\Lambda_{s;\,j, a}(t)$ are the desired functions, $s\in J_{p;m}^n$,
$j\in \bZ^n$, $a\in I_p^n$.

Substituting (\ref{eq-73}) into equation (\ref{eq-70}), in view of
of Theorem~\ref{th4.1}, we obtain
$$
\sum_{s\in J_{p;m}^n,j\in \bZ^n,a\in I_p^n}
\Big(\frac{d\Lambda_{s;\,j, a}(t)}{dt}
+\cA(-\widehat{p^{j}}s)\Lambda_{s;\,j, a}(t)\Big)
\Theta_{s;\,j a}^{(m)\times}(x)=0.
$$
The last equation is understood in the weak sense, i.e.,
\begin{equation}
\label{eq-73-1}
\sum_{s\in J_{p;m}^n,j\in \bZ^n,a\in I_p^n}
\Bigl\langle\Big(\frac{d\Lambda_{s;\,j, a}(t)}{dt}
+\cA(-\widehat{p^{j}}s)\Lambda_{s;\,j, a}(t)\Big)
\Theta_{s;\,j a}^{(m)\times}(x),\phi(x)\Bigr\rangle=0,
\end{equation}
for all $\phi \in \Phi(\bQ_p^n)$.
Since according to Lemma~\ref{lem-w-1**}, any test function $\phi \in \Phi(\bQ_p^n)$
is represented in the form of a {\em finite} sum (\ref{wav-9.4=1}), the equality
(\ref{eq-73-1}) implies that
$$
\frac{d\Lambda_{s;\,j, a}(t)}{dt}
+\cA(-\widehat{p^{j}}s)\Lambda_{s;\,j, a}(t)=0, \quad
\forall \, s\in J_{p;m}^n, \, j\in \bZ^n, \, a\in I_p^n,
$$
for all $t\ge 0$.
Solving this differential equation, we obtain
\begin{equation}
\label{eq-74}
\Lambda_{s;\,j, a}(t)
=\Lambda_{s;\,j, a}(0)e^{-\cA(-\widehat{p^{j}}s)t},
\quad s\in J_{p;m}^n, \quad j\in \bZ^n, \quad a\in I_p^n.
\end{equation}

By substituting (\ref{eq-74}) into (\ref{eq-73}) we find a solution of
the Cauchy problem (\ref{eq-70}) in the form
\begin{equation}
\label{eq-75*}
u(x,t)=\sum_{s\in J_{p;m}^n,j\in \bZ^n,a\in I_p^n}
\Lambda_{s;\,j, a}(0)e^{-\cA(-\widehat{p^{j}}s)t}
\Theta_{s;\,j a}^{(m)\times}(x).
\end{equation}
Setting $t=0$, we find
$$
u^{0}(x)=\sum_{s\in J_{p;m}^n,j\in \bZ^n,a\in I_p^n}
\Lambda_{s;\,j, a}(0)\Theta_{s;\,j a}^{(m)\times}(x),
$$
where $u^0\in \Phi'(\bQ_p^n)$ and according to (\ref{wav-9.4=3}),
the coefficients $\Lambda_{s;\,j, a}(0)$ are uniquely determined by
(\ref{wav-9.4=4}) as
\begin{equation}
\label{eq-75*-1}
\Lambda_{s;\,j, a}(0)
=\bigl\langle u^{0}(x),\Theta_{s;\,j a}^{(m)\times}\bigr\rangle,
\quad s\in J_{p;m}^n, \quad j\in \bZ^n, \quad a\in I_p^n.
\end{equation}
The relations (\ref{eq-75*}), (\ref{eq-75*-1}) imply (\ref{eq-75}).
In view of (\ref{wav-9.4=5}), the sum (\ref{eq-75}) is finite on
any test function from the Lizorkin space $\Phi(\bQ_p^n)$.

The theorem is thus proved.
\end{proof}

Theorem~\ref{eq-th5} and Corollary~\ref{cor5} imply the following assertion.

\begin{Corollary}
\label{eq-cor5}
The Cauchy problem {\rm (\ref{eq-70-c})} has a unique solution
\begin{equation}
\label{eq-75-c}
u(x,t)=\sum_{s\in J_{p;m}^n,j\in \bZ^n,a\in I_p^n}
\bigl\langle u^{0}(x),\Theta_{s;\,j a}^{(m)\times}\bigr\rangle
e^{-p^{\alpha\max_{1\le r\le n}\{m_r-j_r\}}t}\Theta_{s;\,j a}^{(m)\times}(x),
\end{equation}
for $t\ge 0$, where $\Theta_{s;\,j a}^{(m)\times}(x)$ are $n$-dimensional
$p$-adic wavelets {\rm (\ref{62.8})}.
\end{Corollary}

Solutions of the Cauchy problems (\ref{eq-70}) and (\ref{eq-70-c})
describe the diffusion processes in the space $\bQ_p^n$.

If $\cA(-\widehat{p^{j}}s)>0$, according to (\ref{eq-75}), $u(x,t)\to 0$,
as $t\to \infty$. In particular, this fact holds for the solution
of the Cauchy problem (\ref{eq-70-c}).

\begin{Example}
\label{ex1}\rm
Consider the one-dimensional Cauchy problem (\ref{eq-70-c}) for the
initial data
$$
u^{0}(x)=\Omega(|x|_p)=\left\{
\begin{array}{rcl}
1,  &&|x|_p\le 1, \\
0,  &&|x|_p> 1. \\
\end{array}
\right.
$$
Substituting (\ref{62.1}), (\ref{62.7-11}), (\ref{62.7-12}) into (\ref{eq-75-c}),
we obtain a solution of this Cauchy problem:
$$
u(x,t)=\sum_{s\in J_{p;m}}\sum_{j=m}^{\infty}p^{-j}
e^{-p^{\alpha(m-j)}t}\chi_p\big(sp^{j}x\big)
\Omega\big(|p^{j}x|_p\big).
$$
\end{Example}

{\bf (b)} Now we consider the Cauchy problem
\begin{equation}
\label{eq-70.1}
\left\{
\begin{array}{rclrcl}
i\frac{\partial u(x,t)}{\partial t}-A_xu(x,t)&=&0,
&&\text{in} \quad \bQ_p^n\times (0, \ \infty), \medskip \\
u(x,t)&=&u^0(x),
&&\text{in} \quad \bQ_p^n\times \{t=0\}, \\
\end{array}
\right.
\end{equation}
where $u^0\in \Phi'(\bQ_p^n)$ and a pseudo-differential operator operator $A_x$
is given by (\ref{eq-70-op}). In particular, we have the Cauchy problem
\begin{equation}
\label{eq-70.1-c}
\left\{
\begin{array}{rclrcl}
i\frac{\partial u(x,t)}{\partial t}-D^{\alpha}_xu(x,t)&=&0,
&&\text{in} \quad \bQ_p^n\times (0, \ \infty), \medskip \\
u(x,t)&=&u^0(x),
&&\text{in} \quad \bQ_p^n\times \{t=0\}, \\
\end{array}
\right.
\end{equation}
where $D^{\alpha}_x$ is the Taibleson fractional operator (\ref{61**})
with respect to $x$, \ $\alpha\in\bC$.

Using the above results, one can construct a solution of
the Cauchy problems (\ref{eq-70.1}) and  (\ref{eq-70.1-c}).

\begin{Theorem}
\label{eq-th6}
Let a pseudo-differential operator $A_{x}$ in {\rm (\ref{eq-70.1})} be
such that its symbol $\cA(\xi)$ satisfies the condition {\rm (\ref{64.1***})}:
$$
\cA\big(\widehat{p^{j}}(-s+\eta)\big)=\cA\big(-\widehat{p^{j}}s\big),
\qquad \forall \, \eta \in \bZ_p^n,
$$
for any $j\in \bZ^n$, $s\in J_{p;m}^n$.
Then the Cauchy problem {\rm (\ref{eq-70.1})} has a unique solution
\begin{equation}
\label{eq-75.1}
u(x,t)=\sum_{s\in J_{p;m}^n,j\in \bZ^n,a\in I_p^n}
\bigl\langle u^{0}(x),\Theta_{s;\,j a}^{(m)\times}\bigr\rangle
e^{-i\cA(-\widehat{p^{j}}s)t}\Theta_{s;\,j a}^{(m)\times}(x),
\end{equation}
for $t\ge 0$, where $\Theta_{s;\,j a}^{(m)\times}(x)$ are $n$-dimensional
$p$-adic wavelets {\rm (\ref{62.8})}.
\end{Theorem}

\begin{Corollary}
\label{eq-cor6}
The Cauchy problem {\rm (\ref{eq-70.1-c})} has a unique solution
\begin{equation}
\label{eq-75.1-c}
u(x,t)=\sum_{s\in J_{p;m}^n,j\in \bZ^n,a\in I_p^n}
\bigl\langle u^{0}(x),\Theta_{s;\,j a}^{(m)\times}\bigr\rangle
e^{-ip^{\alpha\max_{1\le r\le n}\{m_r-j_r\}}t}\Theta_{s;\,j a}^{(m)\times}(x),
\end{equation}
for $t\ge 0$, where $\Theta_{s;\,j a}^{(m)\times}(x)$ are $n$-dimensional
$p$-adic wavelets {\rm (\ref{62.8})}.
\end{Corollary}

\subsection{Semi-linear equations}\label{s5.2}
Consider the Cauchy problem for the semi-linear pseudo-differential equation:
\begin{equation}
\label{76-sl}
\left\{
\begin{array}{rcl}
\frac{\partial u(x,t)}{\partial t}+A_xu(x,t) + u(x,t)|u(x,t)|^{2m}=0,
&&\text{in} \quad \bQ_p^n\times (0, \ \infty), \medskip \\
u(x,t)=u^0(x), &&\text{in} \quad \bQ_p^n\times \{t=0\}, \\
\end{array}
\right.
\end{equation}
where pseudo-differential operator $A_x$ is given by (\ref{eq-70-op}), $m\in \bN$,
$u(x,t)$ is the desired distribution such that $u(x,t)\in \Phi'(\bQ_p^n)$
for any $t\ge 0$.

According to Proposition~\ref{pr-w-2**}, a distribution $u(x,t)$
can be realized as an infinite sum of the form
\begin{equation}
\label{77-sl}
u(x,t)=\sum_{s\in J_{p;m}^n,j\in \bZ^n,a\in I_p^n}\Lambda_{s;\,j, a}(t)
\Theta_{s;\,j a}^{(m)\times}(x),
\end{equation}
where $\Lambda_{s;\,j, a}(t)$ are the desired functions,
$\Theta_{s;\,j a}^{(m)\times}(x)$ are elements of the wavelet basis (\ref{62.8}).
We will solve the Cauchy problem in a particular class of distributions $u(x,t)$
such that in representation (\ref{77-sl})
\begin{equation}
\label{82.0-1-sl}
\widehat{p^{j'-j}}a-a' \notin \bZ_p^n, \quad \text{if}\quad j_k<j_k', \quad
k=1,\dots,n.
\end{equation}
In view of (\ref{62.4-1}), in this case all sets
$\{x\in \bQ_p^n:|\widehat{p^{j}}x-a|_p\le 1\}$,
$\{x\in \bQ_p^n:|\widehat{p^{j'}}x-a'|_p\le 1\}$
are disjoint.

\begin{Theorem}
\label{th4.2-sl}
Let a pseudo-differential operator $A_{x}$ in {\rm (\ref{76-sl})} be
such that its symbol $\cA(\xi)$ satisfies the condition {\rm (\ref{64.1***})}:
$$
\cA\big(\widehat{p^{j}}(-s+\eta)\big)=\cA\big(-\widehat{p^{j}}s\big),
\qquad \forall \, \eta \in \bZ_p^n,
$$
for any $j\in \bZ^n$, $s\in J_{p;m}^n$.
Then in the above-mentioned class of distributions {\rm(\ref{77-sl})}, {\rm(\ref{82.0-1-sl})}
the Cauchy problem {\rm(\ref{76-sl})} has a unique solution
$$
u(x,t)
\qquad\qquad\qquad\qquad\qquad\qquad\qquad\qquad\qquad\qquad\qquad\qquad\qquad
$$
\begin{equation}
\label{82.0-sl}
\sum_{s\in J_{p;m}^n,j\in \bZ^n,a\in I_p^n}
\frac{\langle u^{0}(x),\Theta_{s;\,j a}^{(m)\times}\rangle\big(\cA(-\widehat{p^{j}}s)\big)^{1/2m}
e^{-\cA(-\widehat{p^{j}}s)t}\,\Theta_{s;\,j a}^{(m)\times}(x)}
{\big(\cA(-\widehat{p^{j}}s)+\langle u^{0}(x),\Theta_{s;\,j a}^{(m)\times}\rangle^{2m}p^{-m|j|}
\big(1-e^{-2m\cA(-\widehat{p^{j}}s)t}\big)\big)^{1/2m}}
\end{equation}
for $t\ge 0$, where $\Theta_{s;\,j a}^{(m)\times}(x)$ are $n$-dimensional
$p$-adic wavelets {\rm (\ref{62.8})}. Moreover, this formula is applicable
for the case $\cA\equiv 0$.
\end{Theorem}

\begin{proof}
Since $|\chi_p\big(p^{-1}k\cdot(p^{j}x-a)\big)|=1$, taking into account formulas
(\ref{77-sl}), (\ref{82.0-1-sl}), (\ref{62.8}), we obtain
$$
|u(x,t)|^{2}=\sum_{j\in \bZ,k\in J_{p \, 0}^n,a\in I_p^n}
\Lambda_{s;\,j, a}^2(t)p^{-|j|}\Omega\big(|\widehat{p^{j}}x-a|_p\big)
$$
and
\begin{equation}
\label{77.3-sl}
u(x,t)|u(x,t)|^{2m}=\sum_{j\in \bZ,k\in J_{p \, 0}^n,a\in I_p^n}
\Lambda_{s;\,j, a}^{2m+1}(t)p^{-m|j|}\Theta_{s;\,j a}^{(m)\times}(x),
\end{equation}
where the indexes in the above sums satisfy the condition (\ref{82.0-1-sl}).

Substituting (\ref{77.3-sl}) and (\ref{77-sl}) into (\ref{76-sl}),
in view of of Theorem~\ref{th4.1}, we find that
$$
\sum_{j\in \bZ,k\in J_{p \, 0}^n,a\in I_p^n}
\Big(\frac{d\Lambda_{s;\,j, a}(t)}{dt}
+\cA\big(-\widehat{p^{j}}s\big)\Lambda_{s;\,j, a}(t)
\qquad\qquad\qquad\qquad\qquad\qquad
$$
\begin{equation}
\label{78-sl}
\qquad\qquad
+p^{-m|j|}\Lambda_{s;\,j, a}^{2m+1}(t)\Big)\Theta_{s;\,j a}^{(m)\times}(x)=0,
\end{equation}
where the last equation is understood in the weak sense.
Since, according to Lemma~\ref{lem-w-1**}, any test function $\phi \in \Phi(\bQ_p^n)$
is represented in the form of a {\em finite} sum (\ref{wav-9.4=1}), the equality
(\ref{78-sl}) implies that
$$
\frac{d\Lambda_{s;\,j, a}(t)}{dt}
+\cA\big(-\widehat{p^{j}}s\big)\Lambda_{s;\,j, a}(t)
\qquad\qquad\qquad\qquad\qquad\qquad\qquad\qquad
$$
\begin{equation}
\label{78.1-sl}
+p^{-m|j|}\Lambda_{s;\,j, a}^{2m+1}(t)=0, \quad
\forall \, s\in J_{p;m}^n, \, j\in \bZ^n, \, a\in I_p^n,
\end{equation}
for all $t\ge 0$.
Integrating (\ref{78.1-sl}), we obtain
$$
\frac{\Lambda_{s;\,j, a}^{2m}(t)}
{\cA(-\widehat{p^{j}}s)+p^{-m|j|}\Lambda_{s;\,j, a}^{2m}(t)}
=E_{s;\,j, a}e^{-2m\cA(-\widehat{p^{j}}s)t},
$$
i.e.,
\begin{equation}
\label{81.0-sl}
\Lambda_{s;\,j, a}(t)=\frac{E_{s;\,j, a}^{1/2m}\big(\cA(-\widehat{p^{j}}s)\big)^{1/2m}
e^{-\cA(-\widehat{p^{j}}s)t}}
{\big(1-E_{s;\,j, a}p^{-m|j|}e^{-2m\cA(-\widehat{p^{j}}s)t}\big)^{1/2m}},
\end{equation}
where $E_{s;\,j, a}$ is a constant, $s\in J_{p;m}^n$, $j\in \bZ^n$, $a\in I_p^n$.
Substituting (\ref{81.0-sl}) into (\ref{77-sl}), we find
a solution of the problem (\ref{76-sl})
\begin{equation}
\label{81.0-2-sl}
u(x,t)=\sum_{s\in J_{p;m}^n,j\in \bZ^n,a\in I_p^n}
\frac{E_{s;\,j, a}^{1/2m}\big(\cA(-\widehat{p^{j}}s)\big)^{1/2m}
e^{-\cA(-\widehat{p^{j}}s)t}}
{\big(1-E_{s;\,j, a}p^{-m|j|}e^{-2m\cA(-\widehat{p^{j}}s)t}\big)^{1/2m}}
\Theta_{s;\,j a}^{(m)\times}(x),
\end{equation}
$x\in \bQ_p^n$, $t\ge 0$. Setting in (\ref{81.0-2-sl}) $t=0$,
we obtain that
$$
u^{0}(x)=\sum_{s\in J_{p;m}^n,j\in \bZ^n,a\in I_p^n}
\bigg(\frac{E_{s;\,j, a}\cA(-\widehat{p^{j}}s)}
{1-E_{s;\,j, a}p^{-m|j|}}\bigg)^{1/2m}\Theta_{s;\,j a}^{(m)\times}(x),
$$
where $u^0\in \Phi'(\bQ_p^n)$. Hence, according to (\ref{wav-9.4=3}),
the coefficients $E_{s;\,j, a}$ are uniquely determined by
(\ref{wav-9.4=4}) as
$$
\bigg(\frac{E_{s;\,j, a}\cA(-\widehat{p^{j}}s)}
{1-E_{s;\,j, a}p^{-m|j|}}\bigg)^{1/2m}
=\bigl\langle u^{0}(x),\Theta_{s;\,j a}^{(m)\times}\bigr\rangle.
$$
The last equation implies that
$$
E_{s;\,j, a}=\frac{\langle u^{0}(x),\Theta_{s;\,j a}^{(m)\times}\rangle^{2m}}
{\cA(-\widehat{p^{j}}s)+p^{-m|j|}\langle u^{0}(x),\Theta_{s;\,j a}^{(m)\times}\rangle^{2m}}
$$
Substituting $E_{s;\,j, a}$ into (\ref{81.0-2-sl}), we obtain (\ref{82.0-sl}).
In view of (\ref{wav-9.4=5}), the sum (\ref{82.0-sl}) is finite on any test function
from the Lizorkin space $\Phi(\bQ_p^n)$.

Now by passing to the limit as $\cA\to0$ in formula (\ref{82.0-sl}),
one can easily see that this formula (\ref{82.0-sl}) is applicable
for the case $\cA\equiv 0$.

The theorem is thus proved.
\end{proof}

\section*{Acknowledgments}

The authors are greatly indebted to M.~A.~Skopina for fruitful discussions.


\begin{thebibliography}{10}

\bibitem{Al-Ev-Sk}
S.~Albeverio, S.~Evdokimov, M.~Skopina,
$p$-Adic multiresolution analysis and wavelet frames, (2008),
Preprint at the url: http://arxiv.org/abs/0802.1079v1

\bibitem{Al-Kh-Sh3}
S.~Albeverio, A.~Yu.~Khrennikov, V.~M.~Shelkovich,
Harmonic analysis in the $p$-adic Lizorkin spaces:
fractional operators, pseudo-differential equations,
$p$-adic wavelets, Tauberian theorems,
Journal of Fourier Analysis and Applications,
Vol. 12, Issue 4,  (2006), 393--425.

\bibitem{Al-Kh-Sh4}
S.~Albeverio, A.~Yu.~Khrennikov, V.~M.~Shelkovich,
Pseudo-differential operators in the $p$-adic Lizorkin space,
$p$-Adic Mathematical Physics.
2-nd International Conference, Belgrade, Serbia and Montenegro,
15 -- 21 September 2005,
Eds: Branko Dragovich, Zoran Rakic,
Melville, New York, 2006,
AIP Conference Proceedings -- March 29, 2006,
Vol. 826, Issue 1, pp. 195--205.

\bibitem{Al-Kh-Sh5}
S.~Albeverio, A.~Yu.~Khrennikov, V.~M.~Shelkovich,
$p$-Adic semi-linear evolutionary pseudo-differential equations
in the Lizorkin space,
Dokl. Ross. Akad. Nauk, \textbf{415}, no.~3, (2007), 295--299.
English transl. in Russian Doklady Mathematics,
\textbf{76}, no.~1, (2007), 539--543.

\bibitem{Al-Koz}
S.~Albeverio, S.~V.~Kozyrev,
Multidimensional ultrametric pseudodifferential equations, (2007),
http://arxiv.org/abs/0708.2074v1

\bibitem{Ar-Dr-V}
I.~Ya.~Aref$'$eva, B.~G.~Dragovic, and I.~V.~Volovich
On the adelic string amplitudes,
Phys. Lett. {\bf B} \textbf{209} no.~4 (1998), 445--450.

\bibitem{Av-Bik-Koz-O}
V.~A.~Avetisov, A.H.~Bikulov, S.~V.~Kozyrev, and V.~A.~Osipov,
$p$-Adic models of ultrametric diffusion constrained by
hierarchical energy landscapes,
J. Phys. A: Math. Gen. \textbf{12} (2002), 177--189.

\bibitem{Ben-Ben}
J.~J.~Benedetto, and  R.~L.~Benedetto,
A wavelet theory for local fields and related groups,
The Journal of Geometric Analysis \textbf{3} (2004), 423--456.

\bibitem{Ben1}
R.~L.~Benedetto,
Examples of wavelets for local fields,
Wavelets, Frames, and operator Theory,
(College Park, MD, 2003), Am. Math. Soc.,
Providence, RI, (2004), 27--47.

\bibitem{G-Gr-P}
I.~M.~Gel$'$fand, M.~I.~Graev and I.~I.~Piatetskii-Shapiro,
Generalized functions. vol 6: Representation theory and
automorphic functions.
Nauka, Moscow, 1966.

\bibitem{Kh1}
A.~Khrennikov,
$p$-Adic valued distributions in mathematical physics.
Kluwer Academic Publ., Dordrecht, 1994.

\bibitem{Kh2}
A.~Khrennikov,
Non-archimedean analysis: quantum paradoxes, dynamical systems
and biological models.
Kluwer Academic Publ., Dordrecht, 1997.

\bibitem{Kh4}
A.~Khrennikov,
Information dynamics in cognitive, psychological, social
and anomalous phenomena.
Kluwer Academic Publ., Dordrecht, 2004.

\bibitem{Kh-Koz1}
A.~Yu.~Khrennikov, and S.~V.~Kozyrev,
Wavelets on ultrametric spaces,
Applied and Computational Harmonic Analysis
\textbf{19}  (2005), 61--76.

\bibitem{Kh-Koz2}
A.~Yu.~Khrennikov, and S.~V.~Kozyrev,
Pseudodifferential operators on ultrametric spaces and ultrametric
wavelets,
Izvestia Akademii Nauk, Seria Math.
\textbf{69} no.~5 (2005), 133--148.

\bibitem{Kh-Koz3}
A.~Yu.~Khrennikov, and S.~V.~Kozyrev,
Localization in space for free particle in ultrametric quantum mechanics,
Dokl. Ross. Akad. Nauk \textbf{411} no.~3 (2006) 316--322.
English transl. in Russian Doklady Mathematics.
\textbf{74} no.~3 (2006), 906--911.

\bibitem{Kh-Sh1}
A.~Yu.~Khrennikov, V.~M.~Shelkovich,
$p$-Adic multidimensional wavelets and their application to $p$-adic
pseudo-differential operators, (2006),
Preprint at the url: http://arxiv.org/abs/math-ph/0612049

\bibitem{Kh-Sh2}
A.~Yu.~Khrennikov, V.~M.~Shelkovich,
Non-Haar $p$-adic wavelets and pseudo-differential operators,
Dokl.\ Ross.\ Akad.\ Nauk, \textbf{418}, no.~2, 167-170.
English transl. in Russian Doklady Mathematics.
Dokl. Ross. Akad. Nauk, {\bf 418}, no.~2, (2008), 167-170.
English transl. in Russian Doklady Mathematics,
{\bf 77}, no.~1, (2008), 42--45.

\bibitem{Kh-Sh-Sk}
A.~Yu.~Khrennikov, V.~M.~Shelkovich, M.~Skopina,
$p$-Adic refinable functions and MRA-based wavelets, (2007),
to appear in Journal of Approximation Theory, (2008).
See also as Preprint at the url: http://arxiv.org/abs/0711.2820

\bibitem{Kh-Sh-Sm1}
A.~Yu.~Khrennikov, V.~M.~Shelkovich, and O.G.~Smolyanov,
Multiplicative structures in the linear space of
vector-valued distributions,
Dokl. Ross. Akad. Nauk, {\bf 383}, no.~1, (2002), 28--31;
English transl. in Russian Doklady Mathematics.,
{\bf 65}, no.~2, (2002), 169--172.

\bibitem{Kh-Sh-Sm2}
A.~Yu.~Khrennikov, V.~M.~Shelkovich, and O.~G.~Smolyanov,
Locally convex spaces of vector-valued distributions
with multiplicative structures,
Infinite-Dimensional Analysis, Quantum Probability and Related
Topics, {\bf 5}, no.~4, (2002), 1--20.

\bibitem{Koch3}
A.~N.~Kochubei,
Pseudo-differential equations and stochastics over
non-archimedean fields,
Marcel Dekker. Inc. New York, Basel, 2001.

\bibitem{Koch4}
A.~N.~Kochubei,
Fundamental solutions of pseudo-differential equations associated with
$p$-adic quadratic forms,
Russ. Acad. Sci. Izv. Math. {\bf 62}, (1998), 1169-1188.

\bibitem{Koz0}
S.~V.~Kozyrev,
Wavelet analysis as a $p$-adic spectral analysis,
Izvestia Akademii Nauk, Seria Math.
\textbf{66} no.~2 (2002) 149--158.

\bibitem{Koz2}
S.~V.~Kozyrev,
$p$-Adic pseudodifferential operators and $p$-adic wavelets,
Theor.\ Math.\ Physics \textbf{138}, no.~3 (2004), 1--42.

\bibitem{Koz-Os-Av-1}
S.~V.~Kozyrev, V.~Al.~Osipov, V.C.~A.Avetisov,
Nondegenerate ultrametric diffusion,
J. Math. Phys. \textbf{46} no.~6 (2005), 63302--63317.

\bibitem{Liz1}
P.~I.~Lizorkin,
Generalized Liouville differentiation and the
functional spaces $L\sb{p}{}\sp{r}(E\sb{n})$. Imbedding theorems,
(Russian) Mat. Sb. (N.S.) \textbf{60}(102) (1963), 325--353.

\bibitem{Liz3}
P.~I.~Lizorkin,
Operators connected with fractional differentiation, and classes
of differentiable functions, (Russian)
Studies in the theory of differentiable functions of several
variables and its applications, IV. Trudy Mat. Inst. Steklov.
Vol. 117 (1972), 212--243.

\bibitem{Mallat-1}
S.~Mallat,
Multiresolution representation and wavelets,
Ph.~D. Thesis, University of Pennsylvania, Philadelphia, PA, (1988).

\bibitem{Meyer-1}
Y.~Meyer,
Ondelettes et fonctions splines,
S\'eminaire EDP. Paris, (D\'ecembre 1986).

\bibitem{Schaefer}
H.~Schaefer,
Topological vector spaces. New York--London, 1966.

\bibitem{S-Sk-1}
V.~M.~Shelkovich, M.~Skopina
$p$-Adic Haar multiresolution analysis and pseudo-differential operators,
to appear in Journal of Fourier Analysis and Applications, (2008).
See as the Preprint (2007), http://arxiv.org/abs/0705.2294

\bibitem{Taib1}
M.~H.~Taibleson,
Harmonic analysis on $n$-dimensional vector spaces over local
fields. I. Basic results on fractional integration,
Math. Annalen \textbf{176} (1968), 191--207.

\bibitem{Taib3}
M.~H.~Taibleson,
Fourier analysis on local fields.
Princeton University Press, Princeton, 1975.

\bibitem{Vl-V-Z}
V.~S.~Vladimirov, I.~V.~Volovich and E.~I.~Zelenov,
$p$-Adic analysis and mathematical physics.
World Scientific, Singapore, 1994.

\bibitem{Vl-V1}
V.~S.~Vladimirov, I.~V.~Volovich,
$p$-Adic quantum mechanics,
Commun. Math. Phys. \textbf{123} (1989), 659--676.

\bibitem{V2}
I.~V.~Volovich,
$p$-Adic string,
Class. Quant. Grav. \textbf{4} (1987), L83--L87.

\bibitem{Z1}
W.~A.~Zuniga-Galindo,
Pseudo-differential equations connected with $p$-adic forms
and local zeta functions,
Bull. Austral. Math. Soc. \textbf{70} no. 1 (2004), 73--86.

\bibitem{Z2}
W.~A.~Zuniga-Galindo,
Fundamental solutions of pseudo-differential operators
over $p$-adic fields,
Rend. Sem. Mat. Univ. Padova \textbf{109} (2003), 241--245.


\end{thebibliography}
\end{document}